\documentclass[manuscript,screen,authorversion,nonacm]{acmart}

\usepackage{graphicx}
\usepackage{subcaption}
\usepackage{hyperref}
\usepackage{enumerate}

\usepackage{shuffle}
\usepackage{amsmath}

\usepackage{amssymb}
\usepackage{amsthm}
\usepackage{mathtools}
\usepackage{bm}
\usepackage{bbm}
\usepackage{mdframed}
\usepackage{algorithm}
\usepackage{algorithmic}

\usepackage{jlcode}

\usepackage{tikz}
\usepackage{tikz-qtree}
\usetikzlibrary{positioning,arrows.meta}

\usepackage{pgfplots}
\usepgfplotslibrary{groupplots}
\pgfplotsset{/pgfplots/group/horizontal sep=2cm}
\pgfplotsset{/pgfplots/group/vertical sep=2cm}
\pgfkeys{/pgfplots/scale/.style={
  x post scale=#1,
  y post scale=#1,
  z post scale=#1}
}

\newtheorem{theorem}{Theorem}[section]
\newtheorem{proposition}[theorem]{Proposition}
\theoremstyle{definition}

\begin{document}

\title{Fast and General Automatic Differentiation for Finite-State Methods}

\author{Lucas Ondel~Yang}
\correspondingauthor
\email{lucas.ondel@cnrs.fr}
\orcid{0000-0003-4512-0471}
\affiliation{%
  \institution{Université Paris-Saclay, CNRS, Laboratoire Interdisciplinaire des Sciences du Numérique}
  \city{Orsay}
  \country{France}
}

\author{Caio Corro}
\authornote{Authors listed in alphabetical order.}
\affiliation{%
  \institution{INSA Rennes, IRISA, CNRS, Université de Rennes}
  \city{Rennes}
  \country{France}
}

\author{Nicolas Denier}
\authornotemark[1]
\affiliation{%
  \institution{Université Paris-Saclay, CNRS, Laboratoire Interdisciplinaire des Sciences du Numérique}
  \city{Orsay}
  \country{France}
}

\author{Martin Kocour}
\authornotemark[1]
\affiliation{%
  \institution{Brno University of Technology, Faculty of Information Technology}
  \city{Brno}
  \country{Czech Republic}
}

\author{Tina Raissi}
\authornotemark[1]
\affiliation{%
  \institution{Aalto University}
  \city{Aalto}
  \country{Finland}
}

\author{Pablo Riera}
\authornotemark[1]
\affiliation{
  \institution{Departamento de Computación, Universidad de Buenos Aires, Instituto de Ciencias de la Computación, National Scientific and Technical Research Council}
  \city{Buenos Aires}
  \country{Argentina}
}

\renewcommand{\shortauthors}{Ondel~Yang et al.}

\begin{abstract}
We propose a new method, coined the ``morphism-trick'', to implement general vector-Jacobian products of abstract semiring computations in automatic differentiation softwares.
Our approach leads to an efficient and semiring-agnostic backpropagation method for dynamic programming algorithms represented as a finite state automaton.
Formally, the proposed algorithm computes and differentiates the weight of an automaton, \emph{i.e.} the $\oplus$-sum of all accepted paths' weight.
Benchmarks show that, with minimal effort from the user, our method allows computing the gradient of semiring-based recursive functions order of magnitude faster than state-of-the-art automatic differentiation systems.
Implementations are made available via an open-source library distributed under a permissive license.
\end{abstract}


\maketitle

\section{Introduction}
\label{sec:intro}

Most machine learning algorithms rely on the optimization of a cost function. 
When the data exhibits some kind of organisation, it is often desirable to design functions capturing the underlying structure of the data or enforcing \emph{a priori} adequate constraints. The problem or structured prediction (or structured learning) is the problem of designing and optimizing such functions. 

A broad family of structured prediction methods \citep{Goodman1999, Mohri2002Semiring, Rush2020, Ondel2022} is concerned with optimizing functions that are decomposable into a finite number of terms via two abstract operators, a ``multiplication'' ($\otimes$), and an ``addition'' ($\oplus$) forming a semiring.
Such functions are conveniently represented as a directed multigraph, facilitating the exploration of the structure, and they are efficiently evaluated via Dynamic Programming (DP) \cite{bellman1952dp}. 
From a formal perspective, this family of finitely semiring-decomposable functions can be expressed in a unified way via the framework weighted finite state automata \citep{Mohri2002WFST}. This formalism has enabled the development of general, efficient, and scalable softwares \citep{Mohri2000,Riley2009} for DP algorithms for structured prediction problems.
Beyond machine learning applications, the semiring abstraction is also relevant to any graph-related computations and it is a foundation of the the GraphBLAS project \cite{Kepner2016}.  

However, DP algorithms, and semiring-abstracted computations in general, raises major challenges from the perspective of automatic differentiation. Because the semiring is user-defined and unknown beforehand, room for optimization is strongly limited.
As a result, while it is technically possible to automatically differentiate semiring-based computations with tools such as Zygote.jl \citep{Innes2018} or Enzyme.jl \cite{Moses2020}, the computational cost is prohibitive.
This is a major obstacle in the integration of large scale structured prediction methods in end-to-end pipelines which rely on gradient descent for optimisation method.  

This work proposes the following contributions.
\begin{enumerate}
    \item We introduce the ``morphism-trick'': if the semiring $\oplus$-addition is isomorphic to the addition of real numbers, the computation graph to backpropagate the derivatives can be ``flattened'' and an optimized vector-Jacobian product can be implemented without knowing the semiring beforehand.
    \item We show that the underlying assumption---\emph{i.e.} the $\oplus$-addition is isomorphic to the natural addition ``+''---is not restrictive and that the the morphism-trick can be a applied to a broad family of semirings, including multi-valued semirings. 
    We also show that the morphism-trick can be extended to the special case when $\oplus$ is the $\min$/$\max$ operator to automatically differentiate tropical computations.
    \item Using the finite state automata formalism and the morphism-trick, we propose a general algorithm to differentiate efficiently DP problems. Our solution applies equally to finite and infinite recursion problems.
    \item An implementation of the proposed algorithms is provided in \href{https://gitlab.lisn.upsaclay.fr/PTAL/Automata/TensorAutomata.jl}{TensorAutomata.jl} an open-source Julia software.  
\end{enumerate}

\section{Related Works}
\label{sec:relwork}

In \cite{Rush2020}, the authors propose to use semiring-based linear algebra operations with PyTorch's \citep{Paszke2019} AD engine to implement a broad variety of structured learning algorithms.
To alleviate the computational load of direct AD of semiring-based linear algebra operations, they implement optimized functions and the corresponding vector-Jacobian product for a group of carefully chosen semirings.
This necessary optimization limits the flexibility of their approach: a new structured prediction algorithm requiring a different semiring involves rewriting and optimizing key operations and their corresponding vector-Jacobian product.
Finally, since they use exclusively dense operations their approach is not suited for finite state automata which are inherently sparse objects.

Differentiation of finite state automata has been explored in \cite{Hannun2020, Hannun2021} with a main focus on the log-semiring and, to a lesser extent, on the tropical semiring.
In a similar vein, the k2-fsa software \cite{k2} proposes a practical implementation of vector-Jacobian products for finite state automata operations limited to the log-semiring.
In \cite{balakrishnan2024differentiable}, the authors propose to leverage existing AD tools to differentiate through finite state automata computation over arbitrary differentiable semirings.
However, they illustrate their approach on small problems (a 5-states and a 3-states automaton) far from the scale seen in, say, natural language processing. A method for learning automata has been proposed in \cite{Balle2015} but it does not rely on computing the derivatives of the automaton's weights.

Another line of work focuses on latent structures in neural networks \cite{niculae2025latentstructure}, \emph{i.e.} structures that appear as intermediate computation inside a neural network \cite{kim2017structured}.
To this end, several methods have been proposed, including differentiating through regularized dynamic programming algorithms \cite{mensch2018,corro2019latent}, efficient explicit marginalization via sparse distributions \cite{niculae2018,Correia2020}, general frameworks based on perturbed optimization \cite{berthet2020,stewart2023} and implicit differentiation \cite{blondel2022}, among others.

A generalisation of the chain rule based on the semiring abstraction has been proposed in \cite{Du2023}.
The main idea is to replace the ``$+$'' and ``$\times$'' operators to compute the partial derivatives of a function by another semiring $(\oplus, \otimes)$ in order to compute various statistics. 
Our work is opposite: we aim to compute the ``regular'' partial derivatives of a function that is defined as a computation in a semiring $(\oplus, \otimes)$. 

Finally, a concept of derivatives \citep{Brzozowski1964} and partial derivatives \citep{Antimirov1996} for regular expressions and finite state automata has already been defined and extended to weighted finite state automata \citep{Lombardy2005}.
Despite a resemblance with the present work, this is merely homonymy: we aim at computing the rate of change of a function with respect to some parameters whereas their concept of derivatives aims at providing a canonical representation of (weighted) regular expressions.

\section{Preliminaries}
\label{sec:prelim}

In this section, we introduce relevant terms, concepts, and notations used throughout this work. 
We use two formalisms: semirings and weighted finite state automata. 
The concept of semiring allows to work on the family of computations based on binary operators sharing similar properties. 
The concept of weighted finite state automata allows to express the various structures that one can meet in structured prediction problems. 
Together, these two abstractions cast into a general framework DP algorithms: most of them can be seen as computing paths' weight on a cyclic or acyclic finite state machine. 

\subsection{Monoids and Semirings}

\paragraph{Monoid}
A monoid \citep{Perrin1997} $(M, \odot, \epsilon)$ is an algebraic structure over a non-empty set $M$ where $\odot: M \times M \to M$ is an associative binary operation, that is for all $x, y, z \in M$:
\begin{align*}
	(x \odot y) \odot z &= x \odot (y \odot z), \\
\intertext{and $\epsilon \in M$ is a neutral element, that is $\forall x \in M$:}
	x \odot \epsilon &= \epsilon \odot x = x.
\end{align*}
A monoid $M$ is said to be \emph{commutative} if $\odot$ is a commutative operator, \emph{i.e.}, \(\forall x, y \in M: x \odot y = y \odot x\).

\paragraph{Monoid morphism}
A monoid morphism (or simply morphism) is a mapping $\mu$ from a monoid $(M, \odot, \epsilon)$ to another monoid $(M', \bullet, \epsilon')$ such that for all $x, y \in M$:
\begin{align}
	\nonumber \mu(x \odot y) &= \mu(x) \bullet \mu(y) \\
	\nonumber \mu(\epsilon) &= \epsilon'.
\end{align}
A monoid morphism $\mu$ is an \emph{isomorphim} if $\mu$ is bijective.
The monoids $(M, \odot, \epsilon)$ and $(M', \bullet, \epsilon')$ are said to be \emph{isomorphic} if they are linked by an isomorphism.

\paragraph{Semiring}
A semiring \citep{Golan1999} $(S, \oplus, \otimes, \bar{0}, \bar{1})$ is an algebraic structure over a set $S$ where $(S, \oplus, \bar{0})$ is a commutative monoid and $(S, \otimes, \bar{1})$ is a monoid such that $\otimes$ distributes over $\oplus$, that is for all $x, y, z \in S$:
\begin{align*}
	z \otimes (x \oplus y) &= (z \otimes x) \oplus (z \otimes y) \\
	\text{and}~~(x \oplus y) \otimes z &= (x \otimes z) \oplus (y \otimes z).
\end{align*}
Moreover, $\bar 0$ is an absorbing element with respect to $\otimes$, \emph{i.e.}, \(\forall x \in S: x \otimes \bar{0} = \bar{0} \otimes x = \bar{0}\).

\paragraph{Matrix semiring}
The matrix semiring is the semiring formed by the set of square matrices $S^{k \times k}$, with addition as the elementwise $\oplus$-addition of $S$, and  multiplication as the standard matrix multiplication defined with $\oplus$-addition and $\otimes$-multiplication of $S$.
When the context is clear, we use the standard notation, \emph{i.e.}, for $\mathbf{A}, \mathbf{B} \in S^{k \times k}$, the matrix semiring addition is noted $\mathbf{A} + \mathbf{B}$, and the matrix semiring multiplication is noted $\mathbf{A} \mathbf{B}$.
The additive neutral element is the matrix noted $\mathbf{0}$ with all entries set to $\bar{0}$, and the multiplicative neutral element is the matrix noted $\mathbf{I}$ with diagonal elements set to $\bar{1}$ and the others set to $\bar{0}$. We note $\mathbf{A}^n$ the $n$th-iterated multiplication of $\mathbf{A}$ and we define $\mathbf{A}^0 \triangleq \mathbf{I}$.

\subsection{Weighted Finite State Automata}

\paragraph{Weighted Finite State Automata}
A weighted finite state automaton (WFSA) is a 6-tuple $(S, Q, \Sigma, E, \lambda, \rho)$ consisting of:
\begin{itemize}
	\item a semiring $(S, \oplus, \otimes, \bar{0}, \bar{1})$;
	\item a finite set of states $Q$;
	\item a finite set of labels $\Sigma$;
	\item a finite set of transitions $E \subseteq Q \times Q \times \Sigma \times S$;
	\item an initial mapping function $\lambda : Q \rightarrow S$;
	\item a final mapping function $\rho : Q \rightarrow S$.
\end{itemize}

\paragraph{Path}For a transition $e \in E$, we write $o(e)$ its origin state, $d(e)$ its destination state, $\sigma(e)$ its label and $w(e)$ its weight.
A path $\pi = (e_1, ..., e_n) \in E^n$ is a sequence of  transitions such that, for all $1 \leq i < n$, we have $d(e_i) = o(e_{i+1})$.
The set of paths from state $q \in Q$ to state $r \in Q$ is noted $\Pi(q, r)$.
The weight of a path is defined as:
\begin{align*}
    w(\pi) &= w(e_1) \otimes ... \otimes w(e_n),
\end{align*}
and we extend the notation, $o(\pi) = o(e_1)$ and  $d(\pi) = d(e_n)$, to designate respectively the origin state and the destination state of $p$.

\paragraph{Shortest distance}
The shortest-distance between $q \in Q$ and $r \in Q$ is defined as:
\begin{align*}
	D[q,q] &= \bar{1},
	D[q,r] = \bigoplus_{\pi \in \Pi(q,r)} w(\pi).
\end{align*}
The shortest-distance between two states can be evaluated efficiently via DP using the following recursion:
\begin{align}
    \label{eq:SD_DP} D[q, r] = \bigoplus_{s \in \{u \in Q | u \neq q\}} D[q, s] \otimes D[s, r].
\end{align}

\paragraph{Weight of an automaton}
We define the set of initial (resp. final) states $I$ (resp. $F$) as the set of states for which the mapping $\lambda$ (resp. $\rho$) is not $\bar{0}$, \emph{i.e.} $I = \{q \in Q : \lambda(q) \neq \bar{0} \}$ (resp. $F = \{q \in Q | \rho(q) \neq \bar{0} \}$).
The weight of an automaton $\mathcal{A}$ is defined as
\begin{align}
    \label{eq:W_A} \nu(\mathcal{A}) &= \bigoplus_{(q,r) \in I \times F} \lambda(q) \otimes D[q,r] \otimes \rho(q).
\end{align}
The concept of weight of an automaton is central in structured prediction as it formalises the notion of semiring-recursive function representable as a directed multigraph.
Under various disguises, it is the likelihood function of many probabilistic models (\emph{e.g.} Hidden Markov Models \citep{Rabiner1989}, Conditional Random Fields \citep{Lafferty2001}, ...), it is a special instance of the sum-product algorithm in factor graphs \citep{Kschischang2001}, and it is a key computation in most, if not all, structured loss functions for speech processing \citep{Graves2006, Graves2012,Heigold2012,Hoffmeister2012,Hadian2018,Pratap2022,Ehsan2022,Laptev2023}.

\section{Morphism-Trick} 
\label{sec:semiringdiff}

Differentiating through semiring operations is very challenging when the semiring is unknown to the AD system. 
The problem does not come from evaluating the derivatives themselves but rather from scaling the computation.
Even if the function to differentiate is efficiently implemented, AD quickly becomes inefficient if not helped with specialized implementation of relevant vector-Jacobian products. 
We dissect here the steps of differentiating through semiring-based computations to illustrate this issue. 
Our discussion assumes the use of \emph{reverse-mode} differentation \cite{Rall2006,Atilim2018} where the function to differentiate is first evaluated for a given input (\emph{forward} step) and then the derivatives are backpropagated through the computation graph (\emph{backward} step) via iterated vector-Jacobian products.
Since we are interested in computing derivatives, we restrict our discussion to semirings with $\oplus$ and $\otimes$ operators that are \emph{differentiable}, \emph{i.e.} semirings for which the derivatives of $x \oplus y$ (resp. $x \otimes y$) with respect to $x$ and $y$ exist for all $x$ and all $y$. 
This implies that $\oplus$  and $\otimes$ are continuous functions. 
We shall discuss more on this assumption in Section~\ref{sec:semirings}.

\subsection{Automatic Differentiation}
Let $(S, \oplus, \otimes, \bar{0}, \bar{1})$ be a differentiable semiring and we assume the following derivatives
\[
	\frac{\partial x \otimes y}{\partial x},\quad
	\frac{\partial x \otimes y}{\partial y},\quad
	\frac{\partial x \oplus y}{\partial x},\quad\text{and}\quad
	\frac{\partial x \oplus y}{\partial y}
\]
to be computable by the AD system with negligible cost. 

\begin{figure}[t]
	\centering
	\begin{subfigure}{0.35\textwidth}
		\centering
		\begin{tikzpicture}
\Tree[.$\oplus$
\edge[{Stealth}-];
[.{$\otimes$}
\edge[{Stealth}-];
$x_1$
\edge[{Stealth}-];
$y_1$ ]
\edge[{Stealth}-];
[.$\oplus$
\edge[{Stealth}-];
[.{$\otimes$}
\edge[{Stealth}-];
$x_2$
\edge[{Stealth}-];
$y_2$
]
\edge[{Stealth}-];
[....
\edge[{Stealth}-];
[.$\oplus$
\edge[{Stealth}-];
[.$\otimes$
\edge[{Stealth}-];
$x_{K-1}$
\edge[{Stealth}-];
$y_{K-1}$
]
\edge[{Stealth}-];
[.$\otimes$
\edge[{Stealth}-];
$x_K$
\edge[{Stealth}-];
$y_K$
]
]
]
]
]
\end{tikzpicture}
		\caption{Forward pass}
		\label{fig:forward}
	\end{subfigure}%
	\begin{subfigure}{0.6\textwidth}
		\centering
		\begin{tikzpicture}
	\tikzset{every tree node/.style={align=center,anchor=north}}
	\Tree[.$\frac{\partial z}{\partial z_1}$
	\edge[-{Stealth}];
	[.$\frac{\partial z_1}{\partial u_1}$
	\edge[-{Stealth}]; 
	$\frac{\partial u_1}{\partial x_1}$
	\edge[-{Stealth}];
	$\frac{\partial u_1}{\partial y_1}$ ]
	\edge[-{Stealth}]; 
	[.$\frac{\partial z_1}{\partial z_2}$
	\edge[-{Stealth}]; 
	[.$\frac{\partial z_2}{\partial u_2}$
	\edge[-{Stealth}]; 
	$\frac{\partial u_2}{\partial x_2}$
	\edge[-{Stealth}];
	$\frac{\partial u_2}{\partial y_2}$
	]
	\edge[-{Stealth}]; 
	[....
	\edge[-{Stealth}] node[auto=left] {$\frac{\partial z_{...}}{\partial z_{...}}$};
	[.$\frac{\partial z_{K-2}}{\partial z_{K-1}}$
	\edge[-{Stealth}]; 
	[.$\frac{\partial z_{K-1}}{\partial u_{K-1}}$
	\edge[-{Stealth}]; 
	$\frac{\partial u_{K-1}}{\partial x_{K-1}}$
	\edge[-{Stealth}];
	$\frac{\partial u_{K-1}}{\partial y_{K-1}}$
	]
	\edge[-{Stealth}]; 
	[.$\frac{\partial z_{K-1}}{\partial u_{K}}$
	\edge[-{Stealth}]; 
	$\frac{\partial u_{K}}{\partial x_{K}}$
	\edge[-{Stealth}];
	$\frac{\partial u_{K}}{\partial y_{K}}$
	]
	]
	]
	]
	]
\end{tikzpicture}
		\caption{Backward pass}
		\label{fig:backward}
	\end{subfigure}
	
	\caption{
		Illustration of reverse-mode automatic differentiation for the semiring-based dot product $z = (x_1 \otimes y_1) \oplus ... \oplus (x_k \otimes y_k)$.
		The terms $u_i$ are placeholders for the intermediate $\otimes$-products $x_i \otimes y_i$.
		First $z$ is evaluated for a given $\mathbf{x}$ and $\mathbf{y}$ (Fig. \ref{fig:forward}) and then, the partial derivatives are backpropagated via iterated vector-Jacobian products from the root to the leaves of the computation graph (Fig. \ref{fig:backward}).
		The whole process requires building and maintaining the computation graph in memory between the forward and backward computation.
	}
	\label{fig:forward_backward}
\end{figure}

We consider a dot product operation defined for two vectors $\mathbf{x}, \mathbf{y} \in S^K$ as:
\begin{align}
    \label{eq:dot} z = (x_1 \otimes y_1) \oplus (x_2 \otimes y_2) \oplus \dots \oplus (x_K \otimes y_K).
\end{align}
Not only this operation is typical in semiring-based computation, it is also at the heart of DP algorithms as illustrated by \eqref{eq:SD_DP}.  
A concrete sequential implementation of \eqref{eq:dot} requires computing the $\otimes$-products:
\begin{align}
    \label{eq:prod} u_i &= x_i \otimes y_i
\end{align}
and then to perform the $\oplus$-accumulation
\begin{align}
    \label{eq:accum_init} z_{K-1} &= u_{K-1} \oplus u_{K} \\
    \label{eq:accum} z_i &= u_i \oplus z_{i+1}.
\end{align}
The final result is obtained with $z~=~z_1$. 
In this case, assuming constant cost for the $\oplus$ and $\otimes$ operations, the forward computation has a linear complexity $O(K)$.

To compute the partial derivatives $\frac{\partial z}{\partial
x_i}$ one has to unroll the computation by application of the chain rule:
\begin{align}
    \label{eq:backprop} \frac{\partial z}{\partial x_i} &= \underbrace{\frac{\partial z}{\partial z_1} ... \frac{\partial z_{i-1}}{\partial z_i} \frac{\partial z_i}{\partial u_i}}_{\frac{\partial z}{\partial u_i}} \frac{\partial u_i}{\partial x_i}.
\end{align}
An illustration of the computational graph for the forward and backward steps is shown in Fig.~\ref{fig:forward_backward}.

Since we assume reverse-mode AD, \eqref{eq:backprop} is calculated from left to right and corresponds to a top-down traversal of the graph in Figure~\ref{fig:backward}. 
Note that the computation graphs of the forward and backward steps have the same topology.
Therefore, both steps have the same algorithmic complexity\footnote{
    We assume that the $\oplus$ and $\otimes$ operations and their derivatives have the same complexity.
} 
and, presumably, a similar runtime.
In practice, however, the runtime between the forward and backward steps of semiring-based computations tends to differ significantly.

\begin{figure}[t]
	\centering
	\begin{subfigure}{0.45\textwidth}
		\centering
		\includegraphics[scale=0.03]{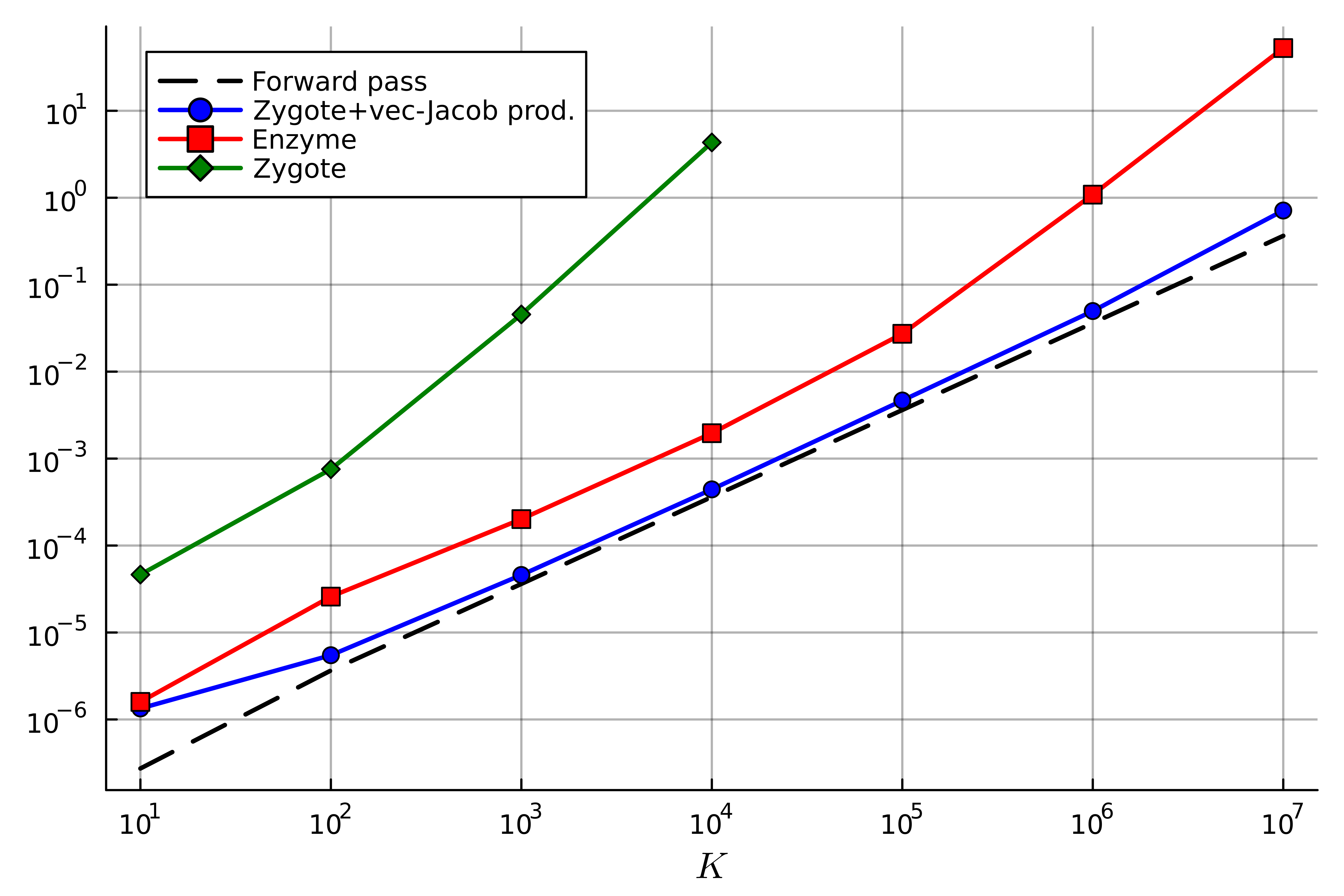}
		\caption{Runtime (s)}
		\label{fig:benchmark_dotproduct_runtime}
	\end{subfigure}%
	\begin{subfigure}{0.45\textwidth}
		\centering
		\includegraphics[scale=0.03]{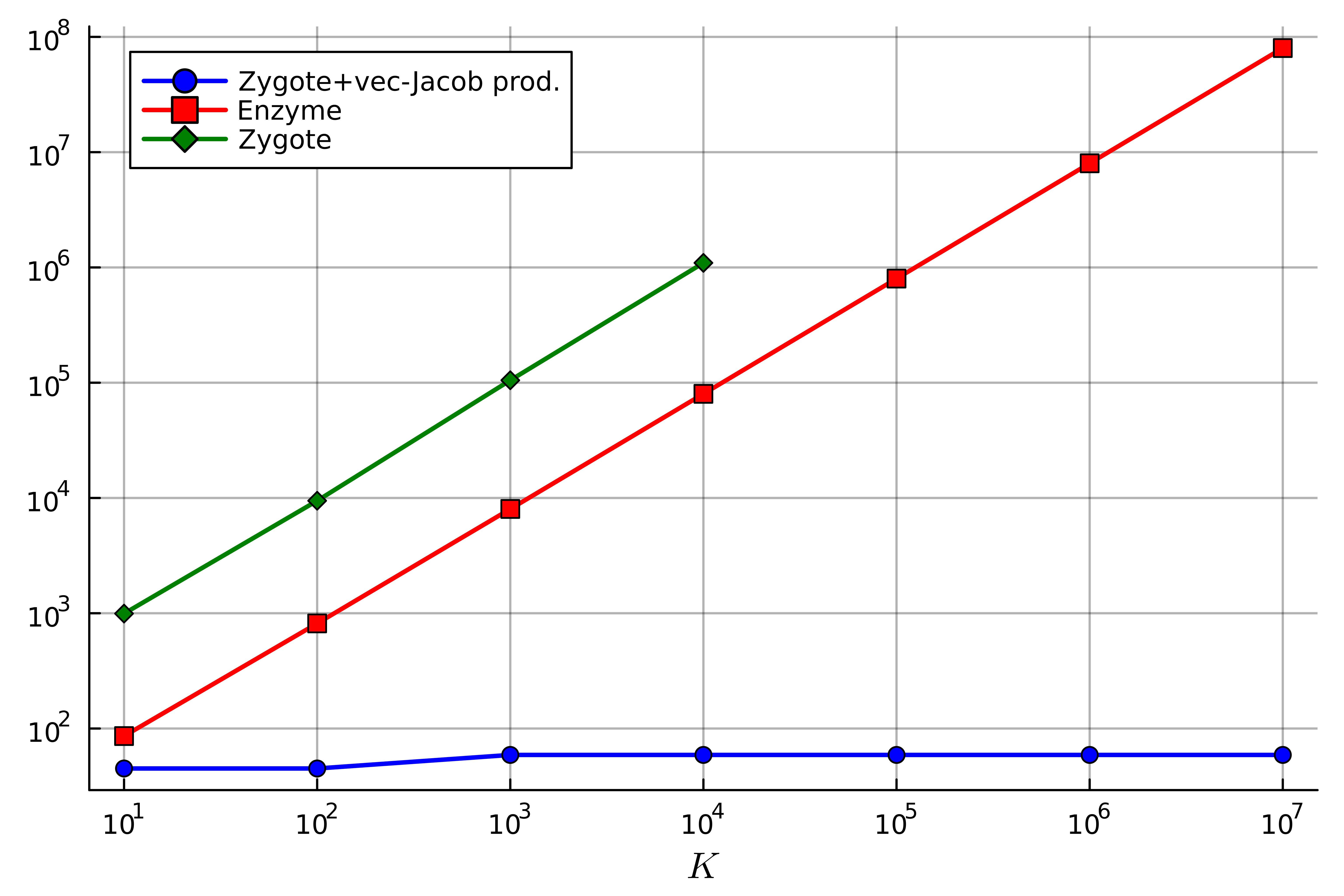}
		\caption{\# Allocations}
		\label{fig:benchmark_dotproduct_allocs}
	\end{subfigure}
	\caption{
		Benchmark comparing the AD of the dot product in the log-semiring with and without the morphism-trick.
		The baselines use Enzyme.jl and Zygote.jl to perform AD while our proposed method relies on Zygote.jl augmented with a general vector-Jacobian product derived from the morphism-trick.
		$K$ is the number of dimension of the input vectors.
		Our approach scales well as it bypasses the needs to allocate and maintain the computational graph in memory.
		Note that the computation of the semiring dot-product itself does not require dynamic memory allocations.
	}
	\label{fig:benchmark_dotproduct}
\end{figure}
We illustrate the issue on the left-hand-side of Figure~\ref{fig:benchmark_dotproduct} where we report the runtime of the forward an backward steps of the dot product under the log-semiring \citep{Droste2009} with $x \oplus y = \log(e^x + e^y)$ and $x \otimes y = x + y$.
For this benchmark, we have used two popular AD softwares of the Julia programming language ecosystem: Zygote.jl \cite{Innes2018} and Enzyme.jl\footnote{In this work, we have used exclusively the Julia binding of Enzyme.} \cite{Moses2020}. 
Zygote.jl benefits from highly optimized implementations of vector-Jacobian products for linear algebra operations (over the field of reals), Enzyme.jl has a minimal set of optimized implementations of vector-Jacobian products, but achieves high-performance by differentiating the optimized low-level LLVM's intermediate representation of the program. 
In both cases the AD process is significantly slower than the forward computation despite having the same theoretical complexity \cite{Griewank2008}. 
This performance gap is explained by the need to store in memory the computation graph necessary for the backward step. 
As one can see on the right-hand-side of Fig.~\ref{fig:benchmark_dotproduct}, the total number of allocations made by the AD system grows with the depth of the computation tree strongly impacting the runtime.

\subsection{Differentiation via the ``Morphism-Trick''}
\label{sec:autodiff_morphism}

We have seen that the performance bottleneck in AD of semiring-based computations is related to the depth of the computation graph, which incurs many unavoidable heap allocations.
In this section we describe a method to ``flatten'' the computation graph for the backward pass by using a unique implementation of vector-Jacobian products for a broad family of semirings.

\begin{proposition}(Morphism-trick)
\label{prop:derivative}
	Let $(S, \oplus, \otimes, \bar{0}, \bar{1})$ be a differentiable semiring such that the monoid $(S, \oplus, \bar{0})$ is isomorphic to $(\mathbb R, +, 0)$ under the morphism $\mu: S \to \mathbb R$ and its inverse $\mu^{-1}: \mathbb{R} \to S$. 
    For all $\mathbf{x}, \mathbf{y} \in S^K$, if 
	\begin{align*}
	    z = (x_1 \otimes y_1) \oplus ... \oplus (x_n \otimes y_n),
	\end{align*}
	then:
	\begin{align*}
    	\frac{\partial z}{\partial x_i} &= \frac{\partial z}{\partial \mu(z)} \frac{\partial \mu(x_i \otimes y_i)}{\partial x_i \otimes y_i} \frac{\partial x_i \otimes y_i}{\partial x_i}.
	\end{align*}
\end{proposition}
\begin{proof} Using the morphism, we have 
\begin{align*}
    z &= \mu^{-1}[\mu(z)] \\
    &= \mu^{-1}[\mu(x_1 \otimes y_1) + ... + \mu(x_n \otimes y_n)]
\end{align*}and applying the chain rule gives:
    \begin{align*}
        \frac{\partial z}{\partial x_i} &= \frac{\partial z}{\partial \mu(z)} \frac{\partial \mu(z)}{\partial x_i \otimes y_i} \frac{\partial x_i \otimes y_i}{\partial x_i}.
    \end{align*}
    Observing that $\frac{\partial \mu(z)}{\partial x_i \otimes y_i}$ simplifies to $\frac{\partial \mu(x_i \otimes y_i)}{\partial x_i \otimes y_i}$ leads to the desired result.
\end{proof}

\begin{figure}[t]
	\centering
    \begin{tikzpicture}
	\tikzset{every tree node/.style={align=center,anchor=north}}
	\Tree[.$\frac{\partial z}{\partial \mu(z)}$
	\edge[-{Stealth}]; 
	[.$\frac{\partial \mu(u_1)}{\partial u_1}$
	\edge[-{Stealth}]; 
	$\frac{\partial u_1}{\partial x_1}$
	\edge[-{Stealth}];
	$\frac{\partial u_1}{\partial y_1}$
	]
	\edge[-{Stealth}];
	...
	\edge[-{Stealth}]; 
	[.$\frac{\partial \mu(u_K)}{\partial u_K}$
	\edge[-{Stealth}]; 
	$\frac{\partial u_K}{\partial x_K}$
	\edge[-{Stealth}];
	$\frac{\partial u_K}{\partial y_K}$
	]
	]
\end{tikzpicture}
	
	\caption{
		Computation graph of the backward pass of $z$ via the morphism-trick.
	}
	\label{fig:backprop_morphism}
\end{figure}

Proposition~\ref{prop:derivative} shows that, under certain conditions, the computation of $\frac{\partial z}{\partial u_i}$ does not require any more unrolling the recursion in \eqref{eq:accum}. 
This implies that the computation graph for the backward step has a depth of $O(1)$ (see Figure~\ref{fig:backprop_morphism}) and no longer requires storing the $z_i$ values saving therefore many heap allocations.
By leveraging the morphism-trick, one can incorporate an optimized implementation of vector-Jacobian products for semiring-based function in an AD system that are valid for any semiring equipped with a morphism $\mu$.
It is sufficient for the end-user to provide to the AD system routines to evaluate the following computations for all $a, b \in S$:
\begin{align}
    \label{eq:generic_drule1} \frac{\partial a}{\partial \mu(a)} \frac{\partial \mu(b)}{\partial b} \text{, } \frac{\partial a \otimes b}{\partial a} \text{, and } \frac{\partial a \otimes b}{\partial b}.
\end{align}
An example of a pseudo-algorithm of the vector-Jacobian product of semiring dot product is provided in Algorithm~\ref{alg:dotrule}.
We encourage the reader to study this algorithm carefully as it is key to understand the benefit of the morphism-trick. 
The first observation is that Algorithm~\ref{alg:dotrule} does not depends on the intermediate values $z_i$ avoiding the overhead of many memory allocations. 
A second important observation is that the semiring-specific part of computation is limited to the three functions defined in \eqref{eq:generic_drule1}. 
In practice, these functions can be provided by the end-user via custom operator overloading, similar to what is usually done to implement the $\oplus$ and $\otimes$ operators.
This mirrors the design of finite state automata libraries, as described in \cite{Mohri2000}, where concepts of object oriented programming help to separate graph-related algorithms from the semiring abstraction of the transitions' weight.
Consequently, the implementation of the vector-Jacobian products for semiring-based operations can be decoupled from the actual semiring used during the forward computation.

For instance, let's assume that the log semiring was used to compute the semiring dot product in the forward pass. 
Note that in this case the morphism $\mu(x) = e^x$ statisfies the constraints in Proposition~\ref{prop:derivative} and is invertible. 
It follows that
\begin{align*}
    \frac{\partial a}{\partial \mu(a)} \frac{\partial \mu(b)}{\partial b} &= \frac{e^b}{e^a} \\
    \frac{\partial a \otimes_{\log} b}{\partial a} &= \frac{\partial a \otimes_{\log} b}{\partial b} = 1.
\end{align*}
Plugging these definitions in Algorithm~\ref{alg:dotrule} allows to backpropagate the derivatives without storing the computation graph of the forward step. 

\begin{algorithm}[t]
	\caption{
		Pseudo-code of the vector-Jacobian product for the semiring dot product operation.
		$\mathbf{x}, \mathbf{y}$ are the inputs of the function and $z$ the output.
		$\Delta z$ is the derivative backpropagated by the AD system.
		The operations $\otimes$, $\frac{\partial z}{\partial \mu(z)} \frac{\mu(u_i)}{u_i}$, $\frac{\partial u_i}{x_i}$, and $\frac{\partial u_i}{y_i}$ are external to the AD system and are provided by the user.
	}
	\label{alg:dotrule}
	\begin{algorithmic}
		\STATE $\Delta\mathbf{x} \gets \mathbf{0}$
		\STATE $\Delta\mathbf{y} \gets \mathbf{0}$
		\FOR{$i \in \{1, ..., K\}$}
		\STATE $u_i = x_i \otimes y_i$ \COMMENT{$u_i$ is reconstructed as it was not stored during the forward step}
		
		\STATE $\Delta x_i \gets \Delta z \cdot \frac{\partial z}{\partial \mu(z)}\frac{\partial \mu(u_i)}{\partial u_i} \cdot \frac{\partial u_i}{\partial x_i}$
		\STATE $\Delta y_i \gets \Delta z \cdot \frac{\partial z}{\partial \mu(z)}\frac{\partial \mu(u_i)}{\partial u_i} \cdot \frac{\partial u_i}{\partial y_i}$
		\ENDFOR
		\RETURN $\Delta \mathbf{x}$, $\Delta \mathbf{y}$
	\end{algorithmic}
\end{algorithm}

In Figure~\ref{fig:benchmark_dotproduct}, we report the runtime of the AD of the dot product in the log-semiring with Zygote.jl augmented with a vector-Jacobian product rule implemented as described in Algorithm~\ref{alg:dotrule}.
This implementation was integrated in Zygote.jl using the ChainRules.jl library\footnote{https://github.com/JuliaDiff/ChainRules.jl}.
We observe a drastic improvement in performance compared to Enzyme.jl and Zygote.jl without dedicated vector-Jacobian product rule. 
In this case, there is no significant difference between the forward step and the backward step: the forward graph is not stored and the overhead of the AD process is negligible.

\section{Applicable Semirings}
\label{sec:semirings}

We have introduced the morphism-trick by assuming $x \oplus y$ to be differentiable with both $x$ and $y$ and by assuming the existence of a morphism $\mu$ and its inverse $\mu^{-1}$. 
These assumptions may appear very limiting but we show now that they allow to address a large family of semirings. 
In this section, we explore formally the scope of application of the morphism-trick and we provide examples of semiring compatible with our framework. 
To help the reader get a sense of the effort needed for a user to implement a new compatible semiring, we also provide implementations in the Julia programming language. 
These examples are simplified excerpts of our differentiable semiring library \href{https://gitlab.lisn.upsaclay.fr/PTAL/Automata/Semirings.jl}{Semirings.jl}.
More details about the software will be given in Section~\ref{sec:software}.
  
\subsection{Semirings Isomorphic to the Real Line}

Let be a differentiable semiring $(S, \oplus, \otimes, \bar{0}, \bar{1})$ defined over a subset of $\mathbb{R}$, $S \subseteq \mathbb{R}$. 
By assuming the existence of the the morphism $\mu$, Proposition~\ref{prop:derivative} implies that the morphism-trick necessitates the $\oplus$-addition to satisfy the functional equation:
\begin{align}
    \label{eq:fs_def} x \oplus y = \mu^{-1}[\mu(x) + \mu(y)].
\end{align}
In \cite{Aczel1948}, Aczél gives the following solution: 
\begin{theorem}
\label{th:associativity}
    Let $U$ be an open or half-open interval in $\mathbb{R}$ and a binary operation $\circ : U \times U \rightarrow U$. 
    If $\circ$ satisfies the following conditions:
    \begin{enumerate}[(I)]
        \item Strict monotonicity: $x \circ y < x' \circ y$ for $x < x'$ (likewise for $y < y'$).
        \item Continuity: $\lim_{x \circ y \to z = a \circ b } (x \circ y) = (\lim_{x \to a} x) \circ (\lim_{y \to b} y)$.
        \item Associativity: $(x \circ y) \circ z = x \circ (y \circ z)$.
    \end{enumerate}
    Then, there exists a stricly increasing continuous function $f: U \rightarrow U$ such that:
    \begin{align}
        \nonumber x \circ y = f^{-1}[ f(x) + f(y) ].
    \end{align}
\end{theorem}
In our case, the $\oplus$-addition naturally fulfills the conditions II and III because $S$ is a differentiable semiring. 
Therefore, the morphism-trick only imposes the restriction that $\oplus$ is a strictly monotone function in both arguments (condition I).
This allows us to state a more accurate version of the morphism-trick for real value semirings:
\begin{proposition}[Morphism-trick 2]
\label{prop:derivative2}
	Let $(S, \oplus, \otimes, \bar{0}, \bar{1})$ be a differentiable semiring over an open (or half-open) interval of $\mathbb{R}$ and the $\oplus$-addition is a strictly monotone function in both arguments. 
	Then, there exists a morphism $\mu$ with inverse $\mu^{-1}$ such that:
    \begin{align}
        \nonumber x \oplus y = \mu^{-1} [ \mu(x) + \mu(y)],
    \end{align}
    and for all $\mathbf{x}, \mathbf{y} \in S^K$, the computation 
	\begin{align*}
	    z = (x_1 \otimes y_1) \oplus ... \oplus (x_n \otimes y_n),
	\end{align*}
    has the following partial derivatives
	\begin{align*}
    	\frac{\partial z}{\partial x_i} &= \bigg[ \frac{\partial \mu(z)}{\partial z}\bigg]^{-1} \frac{\partial \mu(x_i \otimes y_i)}{\partial x_i \otimes y_i} \frac{\partial x_i \otimes y_i}{\partial x_i}.
	\end{align*}
\end{proposition}
\begin{proof} 
    Because $\oplus$ is strictly monotone, continuous, and associative, $\mu$ exists and it is monotonically increasing and it has an inverse (Theorem~\ref{th:associativity}).  
    From Proposition~\ref{prop:derivative} we have:
    \begin{align}
    	\label{eq:morphism_trick_v1} \frac{\partial z}{\partial x_i} &= \frac{\partial z}{\partial \mu(z)} \frac{\partial \mu(x_i \otimes y_i)}{\partial x_i \otimes y_i} \frac{\partial x_i \otimes y_i}{\partial x_i}.
    \end{align}
    Because $\mu$ is invertible the following relation holds for any $z \in S$ 
    \begin{align*}
        \mu^{-1}[ \mu(z) ] = z.
    \end{align*}
    Differentiating with respect to $z$ on both side yields
    \begin{align*}
        \frac{\partial \mu^{-1}[\mu(z)]}{\partial \mu(z)}\frac{\partial \mu(z)}{\partial z}  &= 1 \\
        \implies \frac{\partial \mu^{-1}[\mu(z)]}{\partial \mu(z)} &= \frac{\partial z}{\partial \mu(z)} = \bigg[ \frac{\partial \mu(z)}{\partial z} \bigg]^{-1}.
    \end{align*}
    Note that the inverse is always defined since $\frac{\partial \mu(z)}{\partial z} \neq 0$ due to the strict monotonicity of $\mu$.
    Making use of this identity in \eqref{eq:morphism_trick_v1} completes the proof.
\end{proof}

A typical example statisfying Proposition~\ref{prop:derivative2}, is the family of scaled log-semiring parameterized by $\tau \neq 0$. 
This family is used extensively in statistical methods as it appears in the computation of the log-partition function of many probabilistic models.
\begin{example}[scaled log-semiring]
The scaled log-semiring on the set $\mathbb R$ is defined as:
\begin{align*}
    x \oplus y &= \frac{1}{\tau} \log(e^{\tau x} + e^{\tau y}),
    \\
    x \otimes y &= x + y,
    \\
    \bar 0 &= - \infty \text{~~and~~}\bar 1 = 0.
\end{align*}
The morphism is given by:
\begin{align*}
    \mu(x) &= \exp(\tau x). 
\end{align*}
From Proposition~\ref{prop:derivative2}, the derivative of $z = x \oplus y$ with respect to $x$ is 
\begin{align*}
    \frac{\partial z}{\partial x} = \frac{e^{\tau x}}{e^{\tau z}}
\end{align*}
on the interval $x \in (-\infty, z]$.
An example of implementation of the scaled log-semiring is shown in Figure~\ref{fig:code_log}.
\end{example}

\begin{figure}
\centering
\begin{jllisting}
struct ScaledLogSemiring{τ} <: Semiring
    val
end

Base.:+(x::ScaledLogSemiring{τ}, y::ScaledLogSemiring{τ}) where τ = 
	ScaledLogSemiring{τ}( log(exp(τ*x.val) + exp(τ*y.val)) / τ ) 

Base.:*(x::ScaledLogSemiring{τ}, y::ScaledLogSemiring{τ}) where τ = 
	ScaledLogSemiring{τ}( x.val + y.val )

Base.zero(S::Type{ScaledLogSemiring{τ}}) where τ = 
	ScaledLogSemiring{τ}( τ > 0 ? Inf : -Inf )

Base.one(S::Type{ScaledLogSemiring}) = 0


# Routines for automatic differentiation via the morphism-trick. 
# `Δu` is the "sensitivity" propagated by the (reverse-mode) automatic 
# differentiation engine. 

# Returns: `∂z/∂μ(z) * ∂μ(z)/∂z_i * Δu`.
∂sum(z::ScaledLogSemiring{τ}, z_i::ScaledLogSemiring{τ}, Δu) where τ =
	Δu * exp(τ*(z_i.val - z.val))

# Returns `∂(a ⊗ b)/∂a * Δu`.
∂rmul(a::ScaledLogSemiring, b::ScaledLogSemiring, Δu) = Δu

# Returns `∂(a ⊗ b)/∂b * Δu`.
∂lmul(a::ScaledLogSemiring, b::ScaledLogSemiring, Δu) = Δu
\end{jllisting}
\caption{Example of an implementation of the scaled log-semiring compatible with AD via the morphism-trick.}
\label{fig:code_log}
\end{figure}

Another related example is the log$_\kappa$-semiring derived from $\kappa$-calculus \citep{Kaniadakis2002}. 
This semiring generalizes the classical log-semiring by using the deformed exponential and logarithmic functions:
\begin{align}
    \exp_{\kappa}(x) &= \big(\sqrt{1 + \kappa^2 x^2} + \kappa x \big)^{\frac{1}{\kappa}}, \\
    \log_{\kappa}(x) &= \frac{x^\kappa - x^{-\kappa}}{2\kappa}.
\end{align}
They reduce to the natural exponential and logarithmic functions when $\kappa$ tends to 0.
The log$_\kappa$-semiring is used to define power law distributions.
\begin{example}[log$_\kappa$-semiring]
The log$_\kappa$-semiring on the set $\mathbb R$ is defined as follows:
\begin{align*}
    x \oplus_{\kappa} y &= \log_{\kappa} [ \exp_{\kappa}(x) + \exp_{\kappa}(y)],
    \\
    x \otimes_{\kappa} y &= x\sqrt{1 + \kappa^2 y^2} + y\sqrt{1 + \kappa^2x^2},
    \\
    \bar 0 &= - \infty \text{~~and~~}\bar 1 = 0.
\\
\intertext{The morphism is given by:}
    \mu(x) &= \exp_{\kappa}(x)
\end{align*}
From Proposition~\ref{prop:derivative2}, the derivative of $z = x \oplus y$ with respect to $x$ is 
\begin{align*}
    \frac{\partial z}{\partial x} = \frac{\sqrt{1 + \kappa^2 z^2}}{\sqrt{1 + \kappa^2 x^2}} \cdot \frac{\exp_{\kappa}(x)}{\exp_{\kappa}(z)}.
\end{align*}
on the interval $x \in (-\infty, z]$.
\end{example}

Generally, for any differentiable and strictly monotone function $f$ mapping values from an interval $(a, b)$ to the interval $(\lim_{x \to a} f(x), \lim_{x \to b} f(x))$, the semiring $(S, \oplus, \otimes, \bar{0}, \bar{1})$ defined as:
\begin{align*}
	S &= [ a, b ] \\ 
	x \oplus y &= f^{-1}[ f(x) + f(y) ] \\		
	x \otimes y &= f^{-1}[ f(x) \cdot f(y) ] \\		
	\bar{0} &= a \\
	\bar{1} &= b \\
\end{align*}
statisfies the conditions in Proposition~\ref{prop:derivative2} and the morphism-trick applies.
\begin{example}[sigmoid-semiring]
The semiring over $\mathbb{R} \cup \{-\infty, +\infty \}$ derived from the sigmoid function
\begin{align*}
	\sigma (x) = \frac{1}{1 + e^{-x}},
\end{align*}
and its inverse
\begin{align*}
	\sigma^{-1}(x) = \log \frac{x}{1-x},
\end{align*}
is given by:
\begin{align*}
    x \oplus y &= \sigma^{-1} [ \sigma(x) + \sigma(y) ] \\
	&= \log \Big[ \frac{2 + e^{-x} + e^{-y} + e^{-x-y}}{e^{-x -y} - 1} \Big], \\
    x \otimes y &= \sigma^{-1}[ \sigma(x) \cdot \sigma(y) ] \\
	&= \log \Big( \frac{e^{-x} + e^{-y} + e^{-x-y}}{1+e^{-x} + e^{-y} + e^{-x-y}} \Big), \\
    \bar{0} &= -\infty, \text{~~and~~} \bar{1} = +\infty.
\\
\intertext{The morphism is by definition:}
    \mu(x) &= \sigma(x).
\end{align*}
From Proposition~\ref{prop:derivative2}, the derivative of $z = x \oplus y$ with respect to $x$ is 
\begin{align*}
    \frac{\partial z}{\partial x} = \frac{\sigma(x)(1-\sigma(x))}{\sigma(z)(1 - \sigma(z))}
\end{align*}
on the interval $x \in (-\infty, z]$.
\end{example}

\subsection{Multi-valued Semiring}
It is also possible to extend the morphism-trick to semirings on (a subset of) $\mathbb{R}^n$ when the $\oplus$-addition is defined independently for each dimension, \emph{i.e}:
\begin{align*}
	(x_1, ..., x_n) \oplus (y_1, ..., y_n) &= (x_1 \oplus_1 y_1, ..., x_n \oplus_n y_n),
\end{align*}
and where each $\oplus_i$ operator statisfies the conditions in Theorem~\ref{th:associativity}. 
The morphism $\mu$ is then simply  the aggregation of the relevant morphism $\mu_i$ for each dimension.
A natural example is the semiring of complex numbers $\mathbb{C}$ which finds application in the modelisation of quantum computing\citep{Say2014}.
\begin{example}[complex number-semiring]
The complex number-semiring on the set $\mathbb R^2$ is defined as:
\begin{align*}
	(x, a) \oplus (y, b) &= (x + y, a + b) \\
	(x, a) \otimes (y, b) &= (xy - ab, xb + ay)\\
    	\bar 0 &= (0, 0) \text{~~and~~}\bar 1 = (1, 0).
\\
\intertext{The morphism $\mu$ is the elementwise identity function:}
    \mu[(x, a)] &= (\mu_1(x), \mu_2(a)) = (x, a).
\end{align*}
From Proposition~\ref{prop:derivative2}, the derivative of $(z, c) = (x, a) \oplus (y, b)$ with respect to $(x, a)$ is simply the identity matrix: 
\begin{align*}
	\begin{pmatrix}
		\frac{\partial z}{\partial x} & \frac{\partial c}{\partial x} \\
		\frac{\partial z}{\partial a} & \frac{\partial c}{\partial a} 
	\end{pmatrix} = 
	\begin{pmatrix}
		1 & 0 \\
		0 & 1
	\end{pmatrix}.
\end{align*}
\end{example}

\begin{figure}
\begin{jllisting}
struct EntropySemiring <: Semiring
    logp
    H
end

Base.:+(x::EntropySemiring, y::EntropySemiring) = 
	EntropySemiring( log(exp(x.logp) + exp(y.logp)), x.H + y.H )
Base.:*(x::EntropySemiring, y::EntropySemiring) = 
	EntropySemiring( x.logp + y.logp, exp(x.logp)*y.H + exp(y.logp)*x.H )
Base.zero(::Type{EntropySemiring}) = EntropySemiring(-Inf, 0)
Base.one(::Type{EntropySemiring})  EntropySemiring(0, 0)

# Routines for automatic differentiation via the morphism-trick. 
# `Δu` is the "sensitivity" propagated by the (reverse-mode) automatic 
# differentiation engine. 

# NOTE: Because `EntropySemiring` is a multi-valuded semiring, 
# `Δu` is a tuple `(Δlogp, ΔH)`. 

# Returns: `∂z/∂μ(z) * ∂μ(z)/∂z_i * Δu`.
function ∂sum(z::EntropySemiring, x::EntropySemiring, Δu) 
    Δlogp, ΔH = Δu
    ( (x.logp - z.logp)*Δlogp, ΔH )
end

# Returns `∂(a ⊗ b)/∂a * Δu`.
function ∂rmul(a::EntropySemiring, b::EntropySemiring, Δu) 
    Δlogp, ΔH = Δu
    ( Δlogp + b.H*exp(a.logp)*ΔH, exp(b.logp)*ΔH )
end

# Returns `∂(a ⊗ b)/∂b * Δu`.
function ∂lmul(a::EntropySemiring, b::EntropySemiring{T}, Δu::SemiringTangent)  
    Δlogp, ΔH = Δu
    ( Δlogp + a.H*exp(b.logp)*ΔH, exp(a.logp)*ΔH )
end
\end{jllisting}
\caption{Example of an implementation of the entropy-semiring (logexpectation-semiring) compatible with AD via the morphism-trick.}
\label{fig:code_ent}

\end{figure}

Another example is the expectation semiring \citep{Eisner2002} (and likewise its higher order generalization \citep{Li2009}) which finds application in the computation of the entropy of distributions over rational languages \cite{Cortes2008}.
\begin{example}[log expectation-semiring]
The log expectation-semiring (or entropy-semiring) is defined on $\mathbb{R}^2$ as:
\begin{align*}
    (x, a) \oplus (y, b) &= (\log(e^x + e^y), a + b) \\
    (x, a) \otimes (y, b) &= (x + y, e^x b + e^y a) \\
    \bar 0 &= (-\infty, 0) \text{~~and~~}\bar 1 = (0, 0). \\
\intertext{The morphism $\mu$ is given by:}
    \mu[(x, a)] &= (\mu_1(x), \mu_2(a)) = (\exp(x), a)
\end{align*}
From Proposition~\ref{prop:derivative2}, the derivative of $(z, c) = (x, a) \oplus (y, b)$ with respect to $(x, a)$ is the following matrix: 
\begin{align*}
	\begin{pmatrix}
		\frac{\partial z}{\partial x} & \frac{\partial c}{\partial x} \\
		\frac{\partial z}{\partial a} & \frac{\partial c}{\partial a} 
	\end{pmatrix} = 
	\begin{pmatrix}
		\frac{e^x}{e^z} & 0 \\
		0 & 1
	\end{pmatrix}
\end{align*}
defined on the interval $(x, a) \in (-\infty, z] \times \mathbb{R}$.
An example of implementation of the log expectation-semiring is shown in Figure~\ref{fig:code_ent}.
\end{example}



\subsection{Semirings with $\max$/$\min$ $\oplus$-addition}

A family of semirings that does not fit the conditions of the morphism-trick is the family of idempotent semirings, \emph{i.e.} semirings for wich $x \oplus x = x$ for any $x \in S$. 
This type of semiring typically appears in shortest-path problems \cite{Mohri2002Semiring} and, for this reason, it is of special interest for many graph-related applications. 
The $\oplus$-addition of numeric idempotent semirings is most often defined with the $\min$ function (or equivalently the $\max$ function) and we shall restrict our attention to this particular class of semiring.  
Differentation over $\min$-sum is challenging as the function is not everywhere differentiable. 
We show here that we can generalize the morphism-trick to help to compute the (sub)derivatives of such computations.
 
Consider the $\oplus$-sum
\begin{align*}
    z = z_1 \oplus ... \oplus z_N
\end{align*}
with $z_i \in S$, and let's assume that there exists two mappings $f: S \rightarrow \mathbb{R}$ and $g: S \rightarrow \mathbb{R}$ and $g$ is invertible such that :
\begin{align*}
	z = g^{-1}[ f(z_1) + ... + f(z_N) ]
\end{align*}
In this case, it is easy to see that the partial derivative of $z$ with respect to any term of the sum undergoes a simplification similar to the one used to derive the morphism-trick, namely:
\begin{align}
	\label{eq:pseudo_min_diff} \frac{\partial z}{\partial z_i} &= \frac{\partial z}{\partial g(z)} \frac{\partial f(z_i)}{\partial z_i}. 
\end{align} 
As a matter of fact, when there is a morphism such that $f = g = \mu$, it is exactly the morphism-trick as defined in Proposition~\ref{prop:derivative}.
We draw the reader's attention that we do not assume $f$ and $g$ to be unique for every $z$, another $\oplus$-summation $z' = z'_1 + ... + z'_{N'}$ may require different mappings $f'$ and $g'$.
For semirings with no unique mappings $f$ and $g$, the challenge is to implement a general vector-Jacobian product compatible with the whole family of mappings.

We focus now in the case where $x \oplus y = \min(x,y)$ but the discussion applies equally when $\oplus$ is the $\max$ function.
Let $1_z(x)$ be the indicator function which returns $1$ if $x = z$ and $0$ otherwise.
We define the function $f$ parameterized by $z$ as:
\begin{align*}
	f_{z}(x) &= x \cdot 1_{z}(x),
\end{align*}
and the function $g$ parameterized by $c$ as:
\begin{align*}
	g_c(x) = c x \text{ and its inverse } g_c^{-1}(x) = \frac{x}{c}.
\end{align*}
For each $z~=~\min(z_1, ..., z_n)$ we define $c = \sum_{i = 1}^N 1_z(z_i)$ the number of arguments $z_i$ that are equal to $z$.
Then, we have: 
\begin{align}
    \nonumber z &= z_1 \oplus ... \oplus z_N \\
    \nonumber &= \min(z_1, ..., z_n) \\
    \label{eq:min_sum_mapped} &= g_{c}^{-1}[ f_{z}(z_1) + ... + f_{z}(z_n) ].
\end{align}
Using Equation~\eqref{eq:pseudo_min_diff}, the partial derivatives with respect to $x$ are:
\begin{align}
	\label{eq:pseudo_min_diff2} \frac{\partial z}{\partial z_i} &= \frac{\partial z}{\partial g_c(z)}\frac{\partial f_{z}(z_i)}{\partial z_i} = \frac{1_z(z_i)}{c}.
\end{align}
Because Equation~\eqref{eq:pseudo_min_diff2} has the same form as the derivatives in Proposition~\ref{prop:derivative}, we can use the same implementation of a vector-Jacobian product for the $\min$/$\max$-addition idempotent semirings as for the other semirings fully compatible with the morphism-trick. 
There is however a difficulty as \eqref{eq:pseudo_min_diff2} depends on the value $c$. We shall address this issue shortly, but let's first consider a simple example.
If $z = z_1 \oplus z_2 \oplus z_3$, $z_1 = z_3 = 3$, and $z_2 = 5$.
Then, $z = \min(3, 5, 3) = 3$, $c = 2$, and we have:
\begin{align*}
	z &= g_2^{-1}[ f_{3}(z_1) + f_{3}(z_2) + f_{3}(z_3)] \\
	&= \frac{1}{2} (z_1 \cdot 1_3(z_1) + z_2 \cdot 1_3(z_2) + z_3 \cdot 1_3(z_3)), 
\end{align*}
and the partial (sub)derivatives are:
\begin{align*}
	\frac{\partial z}{\partial z_1} &= \frac{1_3(z_1)}{2} = \frac{1}{2} \\  
	\frac{\partial z}{\partial z_2} &= \frac{1_3(z_2)}{2} = 0 \\
	\frac{\partial z}{\partial z_3} &= \frac{1_3(z_3)}{2} = \frac{1}{2}.
\end{align*}

We have shown that partial derivatives of $\min$-sum (or $\max$-sum) follows a simplification akin to the one of the morphism-trick and, therefore, can benefit from the same optimized vector-Jacobian products. 
This is achieved by assuming that we know ``$z$'', the numerical value resulting from the sum, and ``$c$'' the actual number minimal values in the $\min$-summation. 
``$z$'' is known since we operate in reverse-mode differentiation but ``$c$'' must be computed explicitily during the forward pass. 
We propose to do that using the notion of \emph{extended-semiring}. 

\paragraph{Extended-semiring} Let $(S, \oplus_S, \otimes_S, \bar{0}_S, \bar{1}_S)$ be an idempotent semiring with $\oplus$ being either the $\min$ or $\max$ function. 
We define the \emph{extended} semiring $(S', \oplus_{S'}, \otimes_{S'}, \bar{0}, \bar{1})$ as:
\begin{align*}
	S' &= S \times \mathbb{R}^+ \\
    	(x, c_x) \oplus_{S'} (y, c_y) &= (x \oplus_S y, c_x 1_{x \oplus_S y}(x) + c_y 1_{x \oplus_S y}(y)) \\
    (x, c_x) \otimes (y, c_y) &= (x \otimes_S y, c_x c_y) \\
	\bar{0} &= (\bar{0}_S, 0) \\
	\bar{1} &= (\bar{1}_S, 0). 
\end{align*}
We say that $S'$ is the \emph{extension} of $S$. 

\begin{figure}
\begin{jllisting}
struct ExtendedTropicalSemiring <: Semiring
    val
    count
end

function Base.:+(x::ExtendedTropicalSemiring, y::ExtendedTropicalSemiring)
	m = min(x.val, y.val)
	c = (x.val == m) * x.count + (y.val == m) * y.count 
	ExtendedTropicalSemiring(m, c)
end

Base.:*(x::ExtendedTropicalSemiring, y::ExtendedTropicalSemiring) =
	ExtendedTropicalSemiring( val(x) + val(y), x.count * y.count )

Base.zero(S::Type{ExtendedTropicalSemiring}) = S(Inf, 0)
Base.one(S::Type{ExtendedTropicalSemiring}) = S(0, 1)

# Routines for automatic differentiation via the extended morphism-trick. 
# `Δu` is the "sensitivity" propagated by the (reverse-mode) automatic 
# differentiation engine. 

# Returns: `∂z/∂g_c(z) * ∂f_z(z_i)/∂z_i * Δu`.
∂sum(z::ExtendedTropicalSemiring, x::ExtendedTropicalSemiring, Δu) = 
	(x.val == z.val) * x.count * Δu / z.count

# Returns `∂(a ⊗ b)/∂a * Δu`.
∂rmul(x::ExtendedTropicalSemiring, a::ExtendedTropicalSemiring, Δu) = Δu

# Returns `∂(a ⊗ b)/∂b * Δu`.
∂lmul(a::ExtendedTropicalSemiring, b::ExtendedTropicalSemiring, Δu) = Δu
\end{jllisting}
\caption{
	Example of an implementation of the extended tropical-semiring compatible with AD via the extended morphism-trick.
}
\label{fig:code_trop}
\end{figure}

\begin{proposition}(Extended Morphism-Trick)
\label{prop:extended_morphism_trick}
	Let $(S, \oplus_S, \otimes_S, \bar{0}_S, \bar{1}_S)$ be a $\min$-$\oplus$ (or $\max$-$\oplus$) idempotent semiring, and let be $S'$ the extension of $S$.
For any $\oplus$-summation
\begin{align*}
	z = z_1 \oplus_S ... \oplus_S z_N,
\end{align*}
we define the functions:
\begin{align*}
	f_z(x) &= x \cdot 1_z(x) \\
	g_c(x) &= c x. 
\end{align*}
For a finite set of values  $(x_1, 1), (x_2, 1), ..., (x_N, 1) \in S'$, if
\begin{align*}
	(z, c) = h[(x_1, 1), (x_2, 1), ..., (x_N, 1)]
\end{align*} 
is a function based solely on $\oplus_{S'}$ and $\otimes_{S'}$ (\emph{i.e.} $h$ is a polynomial in $S'$), then the result $(z, c)$ can be decomposed into a finite $\oplus_{S'}$-sum of $K$ terms:
\begin{align}
	\label{eq:z_sum} (z, c) = (z_1, c_1) \oplus_{S'} ... \oplus_{S'} (z_K, c_K), 
\end{align}
and the partial derivative of $z$ with respect to $x_i$ is:
\begin{align}
	\label{eq:z_subderivative} \frac{\partial z}{\partial x_i} = \frac{\partial z}{\partial g_c(z)} \sum_{j=1}^K\frac{\partial f_z}{\partial z_j} \frac{\partial z_j}{\partial{x_i}}.
\end{align}
\end{proposition}
\begin{proof} 
We provide here a sketch of the proof. 
It is done in 3 steps: 
\begin{enumerate}
	\item {
		Show that $(z, c)$ can be expressed as \eqref{eq:z_sum}. This is trivial as it directly follows from the semiring properties.
	}
	\item {
		Show that $c$ is number of computation paths induced by the ``ties'' in the $\min$ (or $\max$) function. 
		This can be proven by induction with the operation of the extended-semiring $S'$. 
		Note that the recursion needs to be initialized with a path count set to $1$ and this is why we used the values $(x_i, 1)$.  	
	}
	\item {
		Show that \eqref{eq:z_subderivative} is indeed a subderivative of $z$ with respect to $x$. Following our previous discussion, we have:
		\begin{align*}
			\frac{\partial z}{\partial x_i} &= \frac{1}{c} \sum_{j=1}^K 1_z(z_i) \frac{\partial z_j}{\partial x_i},
		\end{align*} 
		and noting that we take an average of all derivatives of winning computation complete the proof.
	}
\end{enumerate}
\end{proof}

Proposition~\ref{prop:extended_morphism_trick} shows that $\min$/$\max$-$\oplus$ idempotent semirings can benefit from efficient general vector-Jacobian product implementations even though they cannot be morphed onto the real line. 
However, the forward computation must be done in an extended version of the semiring that accumulates the total number of equivalent paths in the computation tree. 
This includes, among others, the tropical semiring, the arctic semiring, and the \L{}ukasiewicz semiring \citep{Droste2009}.
An example implementation of the extension of the tropical-semiring $(\mathbb{R} \cup -{\infty}, \min, +, -\infty, 0)$ is provided in Figure~\ref{fig:code_trop}.

\section{Differentiating DP Algorithms with Finite State Methods}
\label{sec:weight}

We have introduced the morphism-trick on the semiring dot product: an operation at the core of DP algorithms.
However, focusing on this operation alone is not enough to derive a general implementation of a vector-Jacobian product for DP algorithms. 
In this section, we leverage the morphism-trick to derive a general algorithm for the forward and backward pass of DP algorithms that can be represented as directed multigraphs. 
Generality is achieved by expressing the problem in terms of an arbitrary finite state automaton rather than considering a particular DP algorithm. 
Importantly, we do not make the assumption that the multigraph is acyclic and, therefore, we also consider the problem of infinite recursions converging to a fixed-point solution.

Recall from Section~\ref{sec:prelim} that the weight of an automaton $\mathcal{A}$, noted $\nu(\mathcal{A})$, formalises semiring-recursive functions. 
The computation of $\nu(\mathcal{A})$, as defined in \eqref{eq:SD_DP}, requires exploring the graph of the structural constraints represented by the automaton $\mathcal{A}$.
This is usually done via a queue-like mechanism, where states to explore are added progressively \citep{Mohri2002Semiring}. 
From a software perspective, queue-like mechanisms present major challenges as, (i) they are sequential in nature and cannot easily benefit from parallel optimizations, and (ii) they are difficult to differentiate through as it is necessary to keep track of the internal state of the queue over time to ensure gradient correctness. 

Because of these difficulties, it is common to have task-specific implementation of DP algorithms (and their derivatives) rather than general ones. 
A glaring example is the CTC loss function \citep{Graves2006} widely used in automatic speech recognition: common implementations, as in PyTorch \citep{Paszke2019}, rely upon a fixed structure to derive a parallel and optimized implementation.  
Therefore, exploring a slightly different structure, as in \citep{Pratap2022}, necessitates the development of new complex softwares \cite{Hannun2020}.   

Inspired by modern numerical methods for graphs \citep{Kepner2011}, we present here a general DP algorithm expressed entirely in terms of semiring linear algebra operations, bypassing the difficulty of the queue-based graph exploration. Then, we apply the morphism-trick on this new algorithm to derive an efficient and general vector-Jacobian product.

\subsection{Forward Pass: Computing $\nu(\mathcal{A})$}
Let $\mathcal{A} = (S, Q, \Sigma, E, \lambda, \rho)$ be a finite state automaton.
We define the matrix $\mathbf{T} \in S^{|Q| \times |Q|}$ such that the element at the $i$th row and $j$th column, denoted $T_{ij}$, is the $\oplus$-sum of all the transitions' weight going from state $i$ to state $j$. Noting $(\delta_{ij}, \sigma)$ the weight of the transition from state $i$ to state $j$ with label $\sigma$, the element $T_{ij}$ of the matrix is given by: 
\begin{align}
    \label{eq:SUM_TRANS} T_{ij} = \bigoplus_{\sigma \in \Sigma} (\delta_{ij}, \sigma). 
\end{align}
We further define the vector $\boldsymbol{\alpha}$ (resp. $\boldsymbol{\omega}$), taking value in $S^Q$, such that $\alpha_i = \lambda(i)$ (resp. $\omega_i = \rho(i)$). 

Using $\mathbf{T}$, $\boldsymbol{\alpha}$, and $\boldsymbol{\omega}$, the weight of $\mathcal{A}$, as described in \eqref{eq:W_A}, can be expressed in terms of semiring linear algebra as: 
\begin{align}
	\label{eq:sd} \mathbf{d}^\top &= \boldsymbol{\alpha}^\top \mathbf{T}^*, \\
	\label{eq:weight_linalg} \nu(\mathcal{A}) &= \mathbf{d}^\top \boldsymbol{\omega}, 
\end{align}
where $\mathbf{T}^*$ denotes the star operation in the matrix semiring, \emph{i.e.}, $\mathbf{T}^0 + \mathbf{T}^1 + \mathbf{T}^2 + ...$. 
Intuitively, \eqref{eq:weight_linalg} is the sum of the weights of all the paths of length $0$, $1$, $2$, and so on.
The intermediate computation $\mathbf{d}$ is known as the single-source shortest-distance algorithm \citep{Mohri2002Semiring} and is the bulk of the algorithm's complexity. 
We assume the infinite sum $\mathbf{T}^*$ to be well-defined and to converge. To derive an efficient recursive algorithm for $\mathbf{d}$, we make use of the following proposition (adapted from Theorem 2.17 in \cite{Esik2009}): 
\begin{proposition}\label{prop:LIN_ED_SOLUTION}
	Let $(S, \oplus, \otimes, \bar{0}, \bar{1})$ be a semiring, $k$ a positive integer, $\boldsymbol{\alpha} \in S^k$ a vector, and $\mathbf{T} \in S^{k \times k}$ a matrix such that $\mathbf{T}^*$ is well-defined and $\lim_{n \to +\infty} \sum_{i=0}^n \mathbf{T}^n$ converges to $\mathbf{T}^*$. 
Then, the sequence $(\mathbf{x}_n, n \in \mathbb{N})$, defined as:
\begin{align}
	\nonumber \mathbf{x}_0 &= \boldsymbol{\alpha}, \\
	\label{eq:sd_recursion} \mathbf{x}_{n+1}^\top &= \boldsymbol{\alpha}^\top + \mathbf{x}_n^\top \mathbf{T}, 	
\end{align}
converges to $\mathbf{d}^\top = \boldsymbol{\alpha}^\top \mathbf{T}^* $.
Furthermore, $\mathbf{d}$ is a fixed-point of \eqref{eq:sd_recursion}, \emph{i.e.}: 
	\begin{align}
	   \label{eq:fixed_point} \mathbf{d}^\top = \boldsymbol{\alpha}^\top + \mathbf{d}^\top \mathbf{T}.
	\end{align}
\end{proposition}
\begin{proof}
	By induction, we have $\mathbf{x}_n^\top = \sum_{i=1}^n \boldsymbol{\alpha}^\top \mathbf{T}^n$, and consequently, $\lim_{n \to +\infty} \mathbf{x}_n^\top = \boldsymbol{
\alpha}^\top \mathbf{T}^*$.
    Plugging $\mathbf{d}^\top = \boldsymbol{\alpha}^\top \mathbf{T}^*$ on the right-hand side of \eqref{eq:sd_recursion} gives: 
    \begin{align*}
        \mathbf{d}^\top &= \boldsymbol{\alpha}^\top + (\boldsymbol{\alpha}^\top \mathbf{T}^*) \mathbf{T} \\
        &= \boldsymbol{\alpha}^\top ( \mathbf{I} + \mathbf{T}^1 + \mathbf{T}^2 + ...)
    \end{align*}
    and noting that, by definition, $\mathbf{T}^0 = \mathbf{I}$, completes the proof.
\end{proof}

Equation~\eqref{eq:sd_recursion} in Proposition~\ref{prop:LIN_ED_SOLUTION} gives a simple linear algebra-based recursive algorithm to evaluate $\mathbf{d}$. 
In the worst case,, $\boldsymbol{\alpha}$ has $|Q|$ non-zero values and $\mathbf{T}$ has at most $|E|$ non-zero values. 
Therefore, each step of the recursion takes $O(|Q| + |E|)$ steps and $\mathbf{d}$ requires $O\big( n(|Q| + |E|)\big)$ steps where $n$ is the number of recursion steps needed to converge. 
If $\mathcal{A}$ is acyclic, convergence takes at most $n = |Q|$ steps which leads to a quadratic complexity: in the worst case, $|E| = |Q|\delta$, where $\delta$ is the maximal outer degree in the graph, and we get $O\big(n(|Q| + |E|)\big) = O(|Q|^2 + |Q|^2 \delta)$.

If $\mathcal{A}$ is acyclic and topologically sorted, the matrix $\mathbf{T}$ becomes strictly upper triangular and the $i$th dimension of the fixed-point solution of \eqref{eq:fixed_point} can be computed directly as:
\begin{align}
	\label{eq:RECURSION} d_i &= \alpha_i \oplus \Big( \bigoplus_{j=1}^{i-1} d_j \otimes T_{ji} \Big).
\end{align}
Since $d_i$ depends only on other values $d_j$ when $j < i$,  it can be evaluated recursively starting from $d_1 = \alpha_1$. 
Each element $T_{ij}$ will be multiplied and added exactly one times, and taking into account the addition of the vector $\boldsymbol{\alpha}$, computing $\mathbf{d}$ requires $O(|Q| + |E|)$ operations. 
This matches the complexity of other DP algorithms on topologically sorted directed acyclic graphs \citep{Mohri2002Semiring}. 

\subsection{Backward Pass: Differentiation with the morphism-trick}

We now turn to the problem of computing the partial derivatives of $\nu(\mathcal{A})$ with respect to $\boldsymbol{\alpha}$, $\mathbf{T}$, $\boldsymbol{\omega}$, and the transitions' weight $(\delta_{ij}, \sigma)$. 
This requires to backpropagate the derivatives through the three steps of the computation, namely \eqref{eq:SUM_TRANS}, \eqref{eq:sd}, and the dot product \eqref{eq:weight_linalg}. 

\paragraph{$\partial \boldsymbol{\omega}$} The derivatives with respect to $\boldsymbol{\omega}$ poses no real challenges as it follows direcly from the morphism-trick on the semiring dot-product \eqref{eq:weight_linalg}: 
\begin{align}
	\frac{\partial \nu(\mathcal{A})}{\partial \omega_i} &= \frac{\partial \nu(\mathcal{A})}{\partial \mu [\nu(\mathcal{A})]} \frac{\partial \mu(d_i \otimes \omega_i)}{\partial d_i \otimes \omega_i} \frac{\partial d_i \otimes \omega_i}{\omega_i}.
\end{align}

\begin{algorithm}[t]
    \caption{
        Algorithm to evaluate $\Delta \mathbf{d} = \mathbf{J}^* \frac{\partial \nu(\mathcal{A})}{\partial \mathbf{d}}$ when $\mathcal{A}$ is an acyclic and topologically sorted automaton.
	The inputs are $\boldsymbol{\alpha}$, $\mathbf{T}$, $\boldsymbol{\omega}$, $\mathbf{d}$, and $z$ the numerical result of $\nu(\mathcal{A})$.
    }
    \label{alg:weightrule}
    \begin{algorithmic}
        \STATE $\Delta \mathbf{d} \gets \mathbf{0}$ 
        \FOR{$i$ in [$1$, ..., $|Q|$]}
            \STATE $u_i \gets d_i \otimes \omega_i$ \COMMENT{$d_i \otimes \omega_i$ is reconstructed as it was not stored in the forward pass.}
            \STATE $\Delta d_i \gets \frac{\partial z}{\partial \mu(z)}\frac{\partial \mu(u_i)}{\partial u_i} \frac{\partial u_i}{\partial d_i}$ \COMMENT{Compute $\frac{\partial \nu(\mathcal{A})}{\partial d_i}$.}

		\FOR{$j$ in [$i+1$, ..., $|Q|$]}
			\STATE $J_{ji} \gets \frac{\partial d_j}{\partial \mu(d_j)} \frac{\partial \mu(d_i \otimes T_{ij})}{\partial d_i \otimes T_{ij}} \frac{\partial d_i \otimes T_{ij}}{\partial d_i}$
			\STATE $\Delta d_i \gets \Delta d_i + J_{ji} \Delta d_j$
		\ENDFOR
        \ENDFOR

        \STATE \RETURN $\Delta \mathbf{d}$
    \end{algorithmic}
\end{algorithm}

\paragraph{$\partial \boldsymbol{\alpha}, \partial \mathbf{T}$}
We define the following function:
\begin{align}
	F(\mathbf{d}, \boldsymbol{\alpha}, \mathbf{T}) = \boldsymbol{\alpha}^\top + \mathbf{d}^\top \mathbf{T},
\end{align}
and let be $x$ a parameter in either $\boldsymbol{\alpha}$ or $\mathbf{T}$.
From Proposition~\ref{prop:LIN_ED_SOLUTION}, $\mathbf{d}$ is a fixed-point of $F$ and therefore the partial derivative of $\nu(\mathcal{A})$ with respect to $x$ satisfies the following equation:
\begin{align}
	\label{eq:gradient_fixed_point} \bigg(\frac{\partial \nu(\mathcal{A})}{\partial \mathbf{d}} \bigg)^\top \frac{\partial \mathbf{d}}{\partial x} &= \bigg(\frac{\partial \nu(\mathcal{A})}{\partial \mathbf{d}} \bigg)^\top \bigg[ \mathbf{J}^\top \frac{\partial \mathbf{d}}{\partial \mathbf{x}} + \frac{\partial F(\mathbf{d}, \boldsymbol{\alpha}, \mathbf{T})}{\partial x} \bigg],
\end{align}
where $\mathbf{J}$ is the Jacobian matrix of $F$ evaluated at the fixed-point $\mathbf{d}$ with respect to the input $\mathbf{d}$. 
Using the morphism-trick on \eqref{eq:fixed_point}, an element of the matrix $J_{ij}$ is defined as:
\begin{align}
    \label{eq:gradient_recursion} J_{ij} &= \frac{\partial d_i}{\partial \mu(d_i)} \frac{\partial \mu(d_j \otimes T_{ji})}{\partial d_j \otimes T_{ji}} \frac{\partial d_j \otimes T_{ji}}{\partial d_j}.
\end{align}
From Equation~\eqref{eq:gradient_fixed_point} we obtain the following equality:
\begin{align}
	\frac{\partial \mathbf{d}}{\partial x} &= (\mathbf{I} - \mathbf{J}^\top)^{-1} \frac{\partial F(\mathbf{d}, \boldsymbol{\alpha}, \mathbf{T})}{\partial x},
\end{align}
which gives us the general solution:
\begin{align}
	\label{eq:general_d_solution} \frac{\partial \nu(\mathcal{A})}{\partial x} &= \bigg(\frac{\partial \nu(\mathcal{A})}{\partial \mathbf{d}} \bigg)^\top (\mathbf{I} - \mathbf{J}^\top)^{-1}\frac{\partial F(\mathbf{d}, \boldsymbol{\alpha}, \mathbf{T})}{\partial x}.
\end{align}
It suffices then to compute the following partial derivatives:
\begin{align*}
	\frac{\partial F(\mathbf{d}, \boldsymbol{\alpha}, \mathbf{T})}{\partial \alpha_i} &= \begin{pmatrix}
		0 \\	
		\vdots \\
		\frac{\partial d_i}{\partial \mu(d_i)}\frac{\partial \mu(\alpha_i)}{\partial \alpha_i} \\
		\vdots \\
		0 
	\end{pmatrix}, & 
	\frac{\partial F(\mathbf{d}, \boldsymbol{\alpha}, \mathbf{T})}{\partial T_{ij}} &= \begin{pmatrix}
		0 \\
		\vdots \\
		\frac{\partial d_j}{\partial \mu(d_j)}\frac{\partial \mu(d_{i} \otimes T_{ij})}{\partial d_{i} \otimes T_{ij}} \frac{\partial d_i \otimes T_{ij}}{\partial T_{ij}}\\
		\vdots \\
		0 
	\end{pmatrix},
\end{align*}
to get a specific solution for the derivatives with respect to $\boldsymbol{\alpha}$ and $\mathbf{T}$.

\begin{figure}[t]
\centering
\begin{jllisting}
using TensorAutomata

# Any AD engine compatible with ChainRules.jl is possible.
using Zygote

# Alphabet
Σ = Set{Label{Char}}('a':'c')

# Semiring
S = TropicalSemiring{Float64}

U = compile(
	[(src = 1, dest = 2, label=('a',)),
	(src = 2, dest = 3, label=('b',), weight=-2),
	(src = 2, dest = 3, label=('c',), weight=3)];
	init = [1], final = [3], semiring = S, symbols = (Σ,)
)

V = compile(
	[(src = 1, dest = 2, label=('a',)),
	(src = 2, dest = 1, label=('b',), weight=1),
	(src = 2, dest = 3, label=('a',))];
	init = [1], final = [3], semiring = S, symbols = (Σ,)
)

W = shuffle(U, V) 

# The Tropical semiring is an idempotent semiring, so
# we need to use its extension to be able to use 
# the morphism-trick.
W = mapweights(ExtendedTropicalSemiring, W) 

# Compute the gradient of `weight(W)`.
∇W, = Zygote.gradient(W) do X
	weight(X)
end
\end{jllisting}
\caption{
	Example using TensorAutomata.jl to compute the gradient of the weight function over a shuffle automaton.
	An illustration of the automata $\mathcal{U}$, $\mathcal{V}$, $\mathcal{W}$ and $\nabla \mathcal{W}$ is given in Figure~\ref{fig:fsa_illustration}.
}
\label{fig:code_example}
\end{figure}

Equation~\eqref{eq:general_d_solution} is an analytical solution but it is hardly a practical one.  
Even if the matrix $\mathbf{T}$, and therefore the Jacobian matrix $\mathbf{J}$, is very sparse, the inverse matrix $(\mathbf{I} - \mathbf{J}^\top)^{-1}$ will be dense in general. 
Therefore, computing the inverse is impossible for non-trival size automata.  
However, by using the identity $(\mathbf{I} - \mathbf{J})^{-1} = \mathbf{J}^0 + \mathbf{J}^1 + ... = \mathbf{J}^*$ and Proposition~\ref{prop:LIN_ED_SOLUTION}, we have that
\begin{align*}
	\Delta \mathbf{d} &=(\mathbf{I} - \mathbf{J})^{-1} \frac{\partial \nu(\mathcal{A})}{\partial \mathbf{d}}\\
	 &= \mathbf{J}^* \frac{\partial \nu(\mathcal{A})}{\partial \mathbf{d}},
\end{align*} 
is the stationary solution of the following dynamical system:
\begin{align}
	\label{eq:gradient_stationary} \Delta \mathbf{d}_{n+1} &= \frac{\partial \nu(\mathcal{A})}{\partial \mathbf{d}} + \mathbf{J} \Delta \mathbf{d}_{n}.
\end{align}
Therefore, \eqref{eq:gradient_stationary} is the same recursion as \eqref{eq:sd_recursion} with a key difference: \eqref{eq:sd_recursion} (the forward pass) is computed in the semiring $S$ while \eqref{eq:gradient_stationary} (the backward pass) is computed in the field of real numbers. 

As for the forward pass, a more efficient backpropagation algorithm can be derived when the automaton $\mathcal{A}$ is acyclic and topologically sorted.   
In this case, recall that $\mathbf{T}$ is strictly upper triangular and, therefore, one can see from \eqref{eq:RECURSION} and \eqref{eq:gradient_recursion} that the Jacobian matrix $\mathbf{J}$ will be strictly lower-triangular.
This allows the direct evaluation of $\Delta \mathbf{d}$ using the following recursion:
\begin{align}
    \label{eq:gradient_sd} \Delta d_{i} &= \frac{\partial \nu(\mathcal{A})}{\partial d_i} + \sum_{j=i+1}^{|Q|}  J_{ji} \Delta d_j.
\end{align}
Notice that $\Delta d_i$ depends only on $\Delta d_j$, for all $j > i$, \emph{i.e.} the recursion in \eqref{eq:gradient_sd} goes in reverse order compared to the recursion in \eqref{eq:sd_recursion}.
The overall algorithm to compute $\Delta \mathbf{d}$ when $\mathcal{A}$ is acyclic and topologically sorted is given in Algorithm~\ref{alg:weightrule}. 

Finally, we use the morphism-trick on \eqref{eq:SUM_TRANS} to derive the partial derivatives of $\nu(\mathcal{A})$ with respect to the transitions'weight:
\begin{align*}
	\frac{\partial \nu(\mathcal{A})}{\partial (\delta_{ij}, \sigma)} &=\frac{\partial \nu(\mathcal{A})}{\partial T_{ij}} \frac{\partial T_{ij}}{\partial \mu(T_{ij})} \frac{\partial \mu[(\delta_{ij}, \sigma)]}{\partial (\delta_{ij}, \sigma)},
\end{align*}
where $\frac{\partial \nu(\mathcal{A})}{\partial T_{ij}}$ is given by \eqref{eq:general_d_solution}.

\section{Software}
\label{sec:software}

\begin{figure}[t]
	\centering
	\begin{subfigure}{0.45\textwidth}
		\centering
		\includegraphics[scale=0.03]{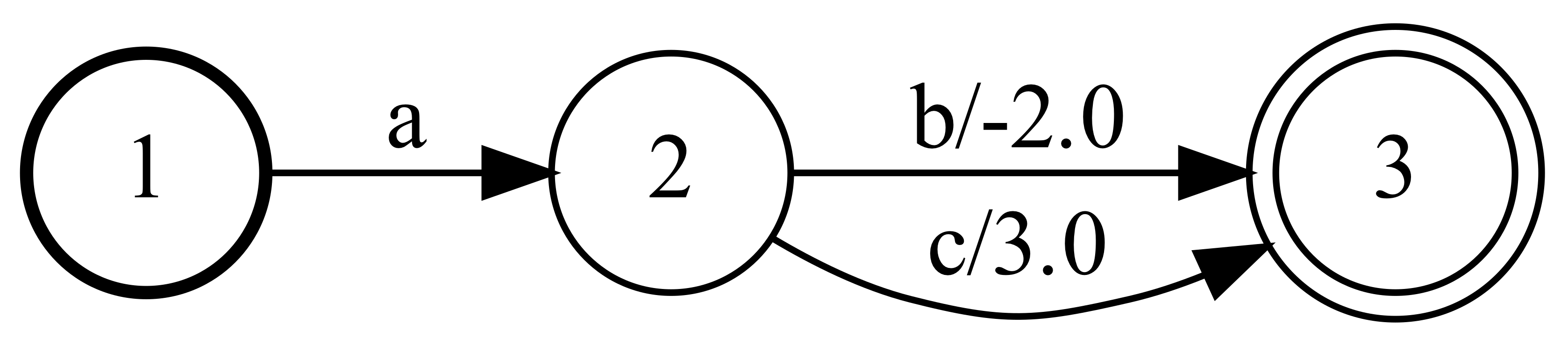}
		\caption{$\mathcal{U}$}
	\end{subfigure}%
	\begin{subfigure}{0.45\textwidth}
		\centering
		\includegraphics[scale=0.03]{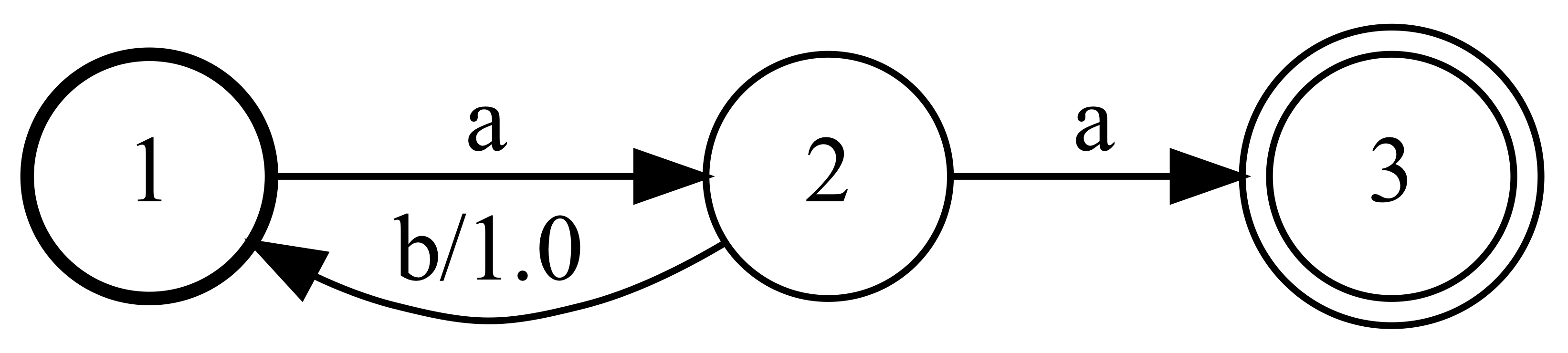}
		\caption{$\mathcal{V}$}
	\end{subfigure}
	\begin{subfigure}{0.9\textwidth}
		\centering
		\includegraphics[scale=0.03]{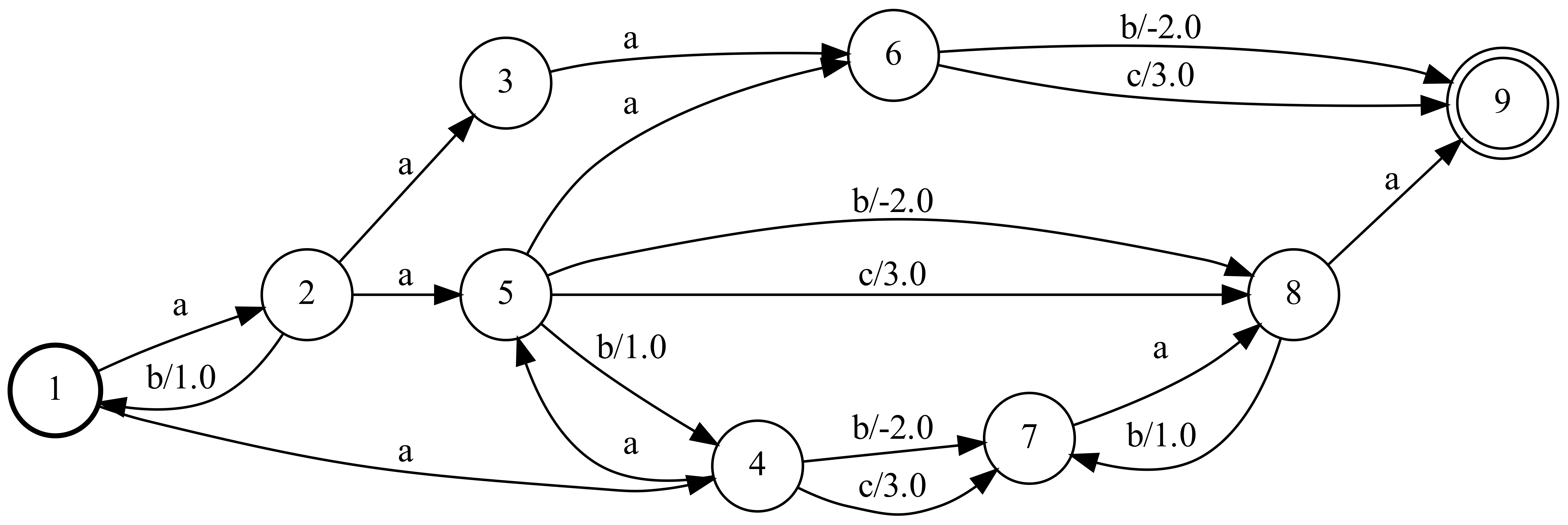}
		\caption{$\mathcal{W} = \mathcal{U} \shuffle \mathcal{V}$}
	\end{subfigure}
	\begin{subfigure}{0.9\textwidth}
		\centering
		\includegraphics[scale=0.03]{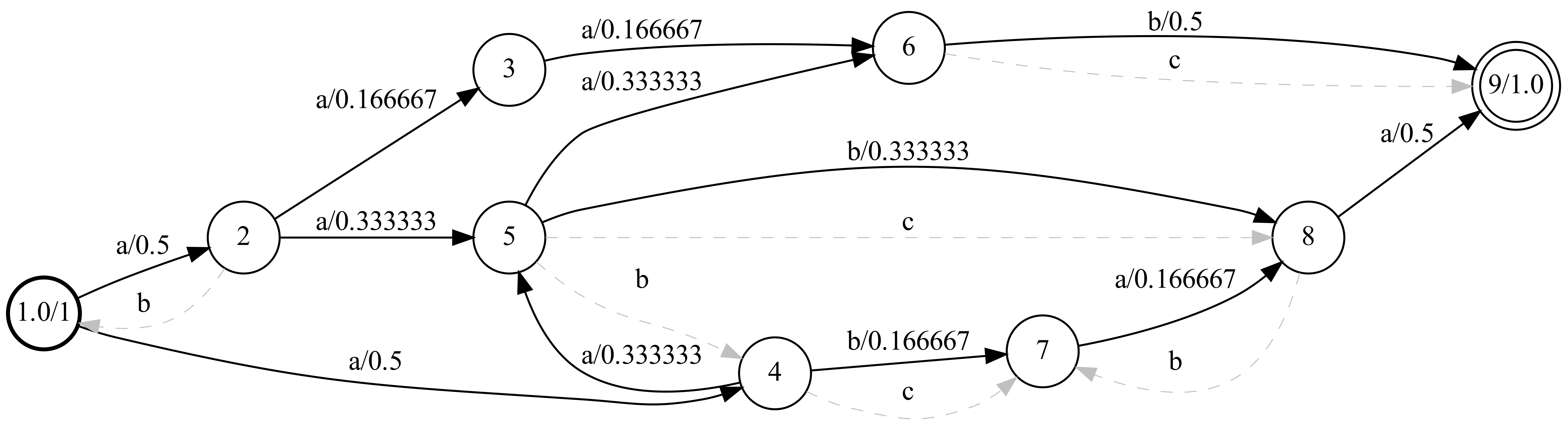}
		\caption{$\nabla \nu(\mathcal{W})$}
	\end{subfigure}
	\caption{
		Representation of the finite state automata constructed in the code example shown in Figure~\ref{fig:code_example}.
		Bold line (resp. double line) states represent initial states (resp. final states). 
		When a transition, an initial state or a final state does not carry weight information, the weight is assumed to be $\bar{1}$. 
	}
	\label{fig:fsa_illustration}
\end{figure}

We provide to the community implementations of the algorithms described in this work under the CeCILL-B license\footnote{\url{https://cecill.info/index.en.html}}, a BSD-like license adapted to French laws: in exchange of strong citations requirements the software can be used without constraints.  
The software is decomposed into three Julia packages: 
\begin{itemize}
	\item {
		\href{https://gitlab.lisn.upsaclay.fr/PTAL/Automata/Semirings.jl}{Semirings.jl}: this package provides an API for implementing custom semirings compatible with the morphism-trick. Some code examples were given in Section~\ref{sec:semirings}. It also provides semirings widely used in the machine-learning community, \emph{e.g.} the log-semiring, tropical-semiring, ... 
	}	
	\item {
		\href{https://gitlab.lisn.upsaclay.fr/PTAL/Automata/ASparseArrays.jl.git}{ASparseArrays.jl}: This package implements sparse linear algebra based on the sparse row compressed storage format \cite{Tinney1967} for finite state automata algorithms. Among other functionalities, it implements the algorithms to compute $\nu(\mathcal{A})$ and the derivatives $\frac{\partial \nu(\mathcal{A})}{\partial x}$ described in Section~\ref{sec:weight}. 
	}	
	\item {
		\href{https://gitlab.lisn.upsaclay.fr/PTAL/Automata/TensorAutomata.jl.git}{TensorAutomata.jl}: This package is a library to manipulate weighted finite state automata. 
		It is build upon Semirings.jl and ASparseArrays.jl and offers an API to design custom automata over arbitrary semirings.
		This package integrates specific vector-Jacobian products to ChainRules.jl\footnote{\url{https://github.com/JuliaDiff/ChainRules.jl}} enabling seamless AD through automata with compatible software (\emph{e.g.} Zygote.jl \cite{Innes2018}).   
	}	
\end{itemize}
These three modules are part of the larger project PTAL\footnote{\url{https://ptal.lisn.upsaclay.fr}} ``Plateforme de Traitement Automatique des Langues'' (Language Processing Platform) which aims to develop machine learning technologies for linguistic studies.  A code example using TensorAutomata.jl to compute the gradient of a shuffle automaton is given in Figure~\ref{fig:code_example}.

\begin{figure}[t]
	\centering
	\begin{subfigure}{0.45\textwidth}
		\centering
		\includegraphics[scale=0.03]{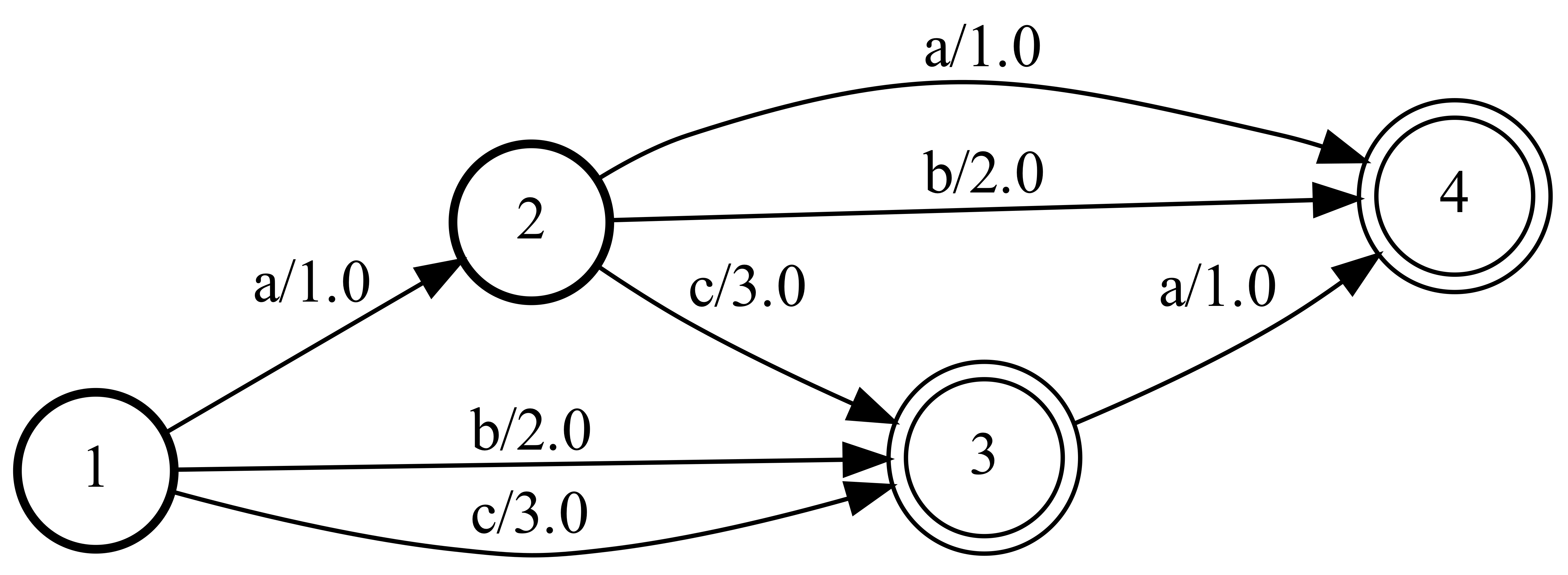}
		\caption{Acyclic topologically sorted automaton.}
		\label{fig:fsa_acyclic_topsorted}
	\end{subfigure}%
	\begin{subfigure}{0.45\textwidth}
		\centering
		\includegraphics[scale=0.03]{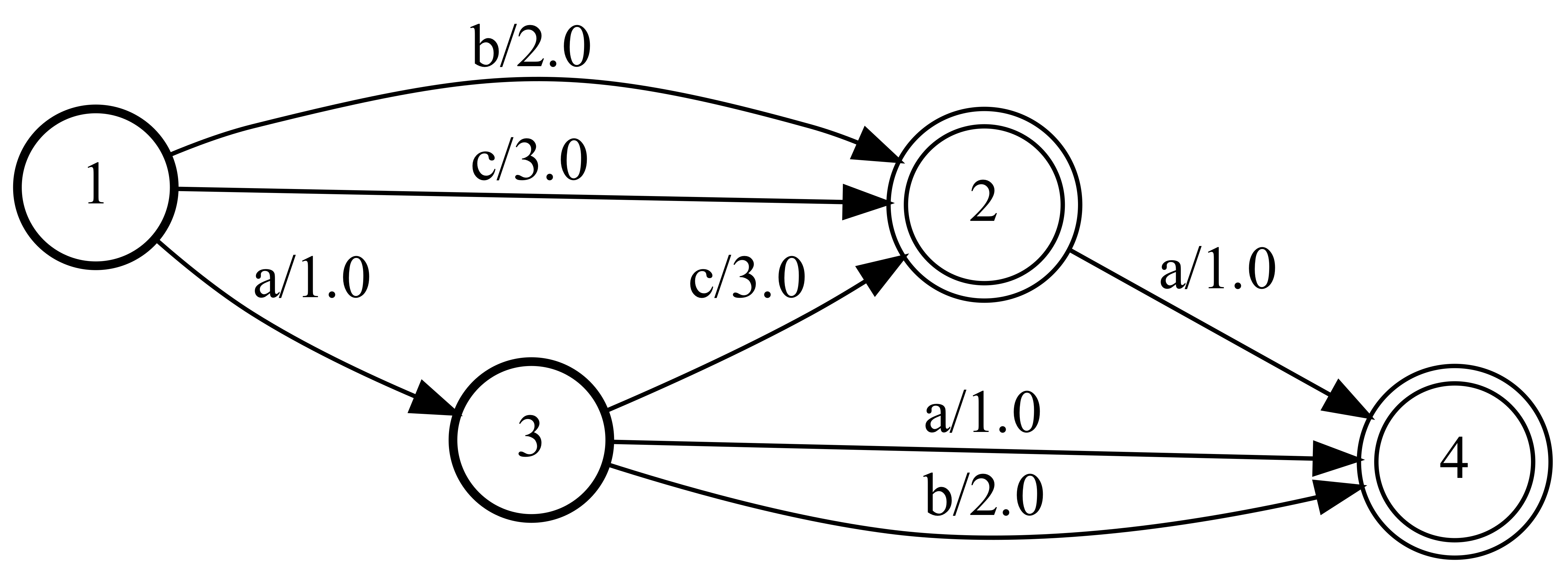}
		\caption{Acyclic automaton.}
		\label{fig:fsa_acyclic}
	\end{subfigure}
	\begin{subfigure}{0.9\textwidth}
		\centering
		\includegraphics[scale=0.03]{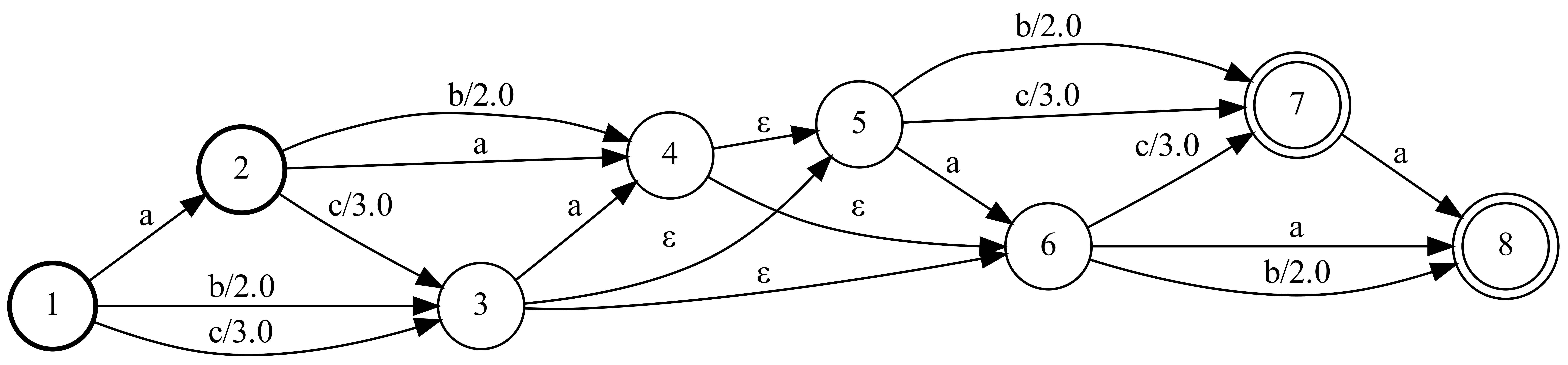}
		\caption{Concatenation of the automaton in Figure~\ref{fig:fsa_acyclic_topsorted} with itself.}
		\label{fig:fsa_acyclic_topsorted_n2}
	\end{subfigure}
	
	\caption{
		Automata used in the benchmark on acyclic and topologically sorted automata. 
		Bold line (resp. double line) states represent intial states (resp. final states). 
		When a transition, an initial state or a final state does not carry weight information, the weight is assumed to be $\bar{1}$. 
		The $\epsilon$ symbol represents the empty string.
	}
	\label{fig:fsa_benchmark}
\end{figure}

\section{Benchmark}

\begin{figure}[t]
	\centering
	\begin{subfigure}{0.45\textwidth}
		\centering
		\includegraphics[scale=0.03]{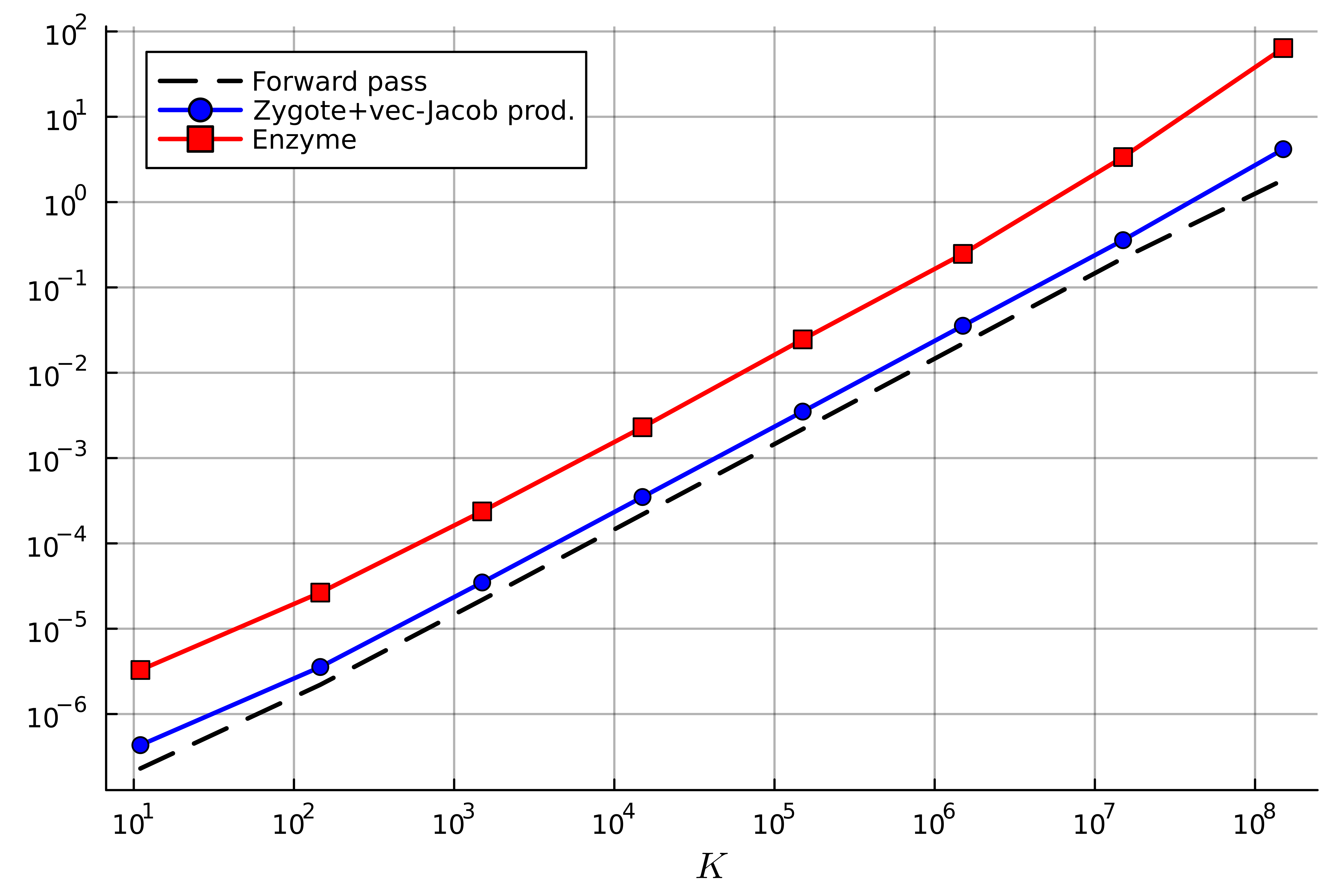}
		\caption{Runtime (s) [acyclic top. sorted].}
		\label{fig:benchmark_acyclic_topsorted_time}
	\end{subfigure}%
	\begin{subfigure}{0.45\textwidth}
		\centering
		\includegraphics[scale=0.03]{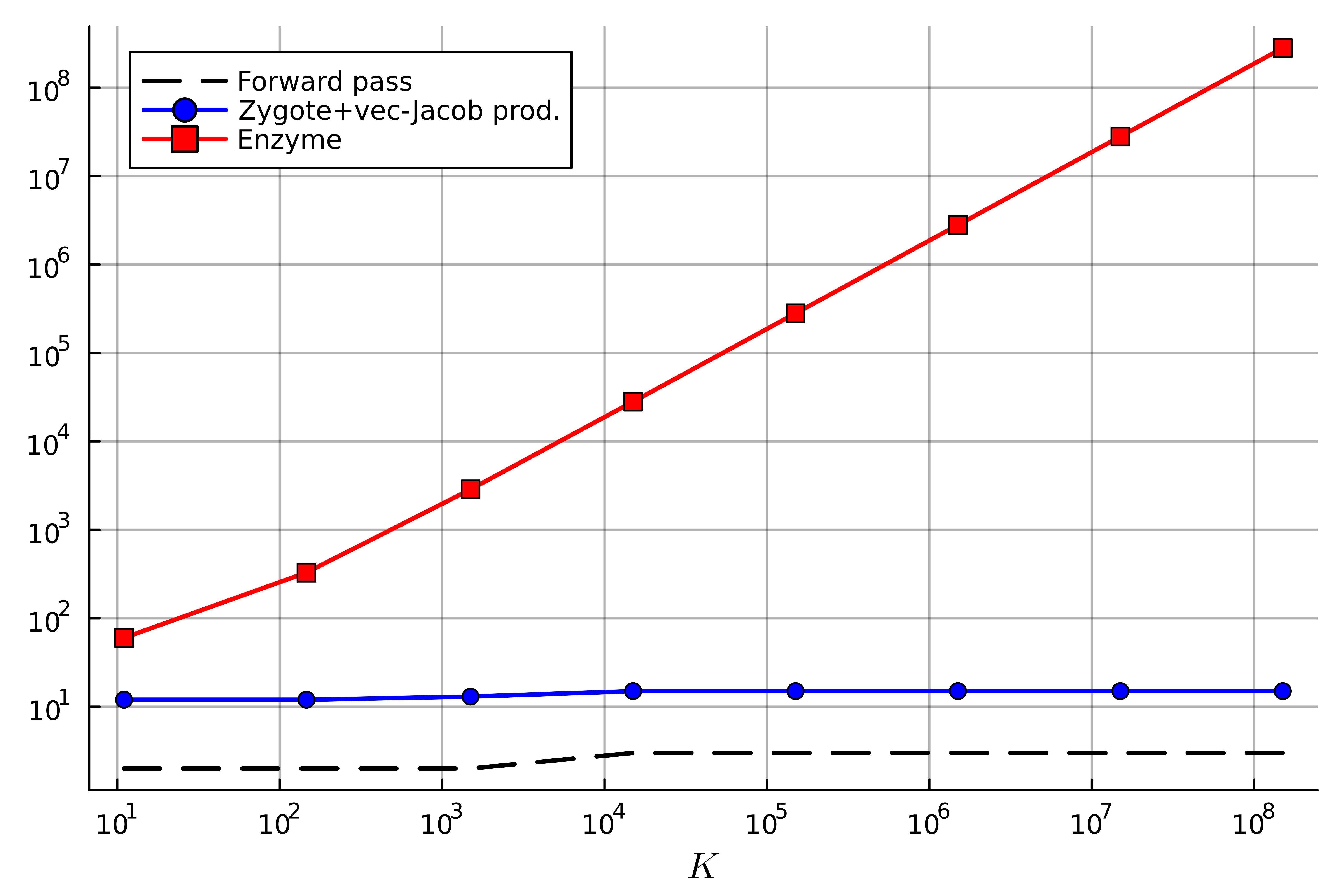}
		\caption{\# Allocations [acyclic top. sorted].}
		\label{fig:benchmark_acyclic_topsorted_allocs}
	\end{subfigure}
	\begin{subfigure}{0.45\textwidth}
		\centering
		\includegraphics[scale=0.03]{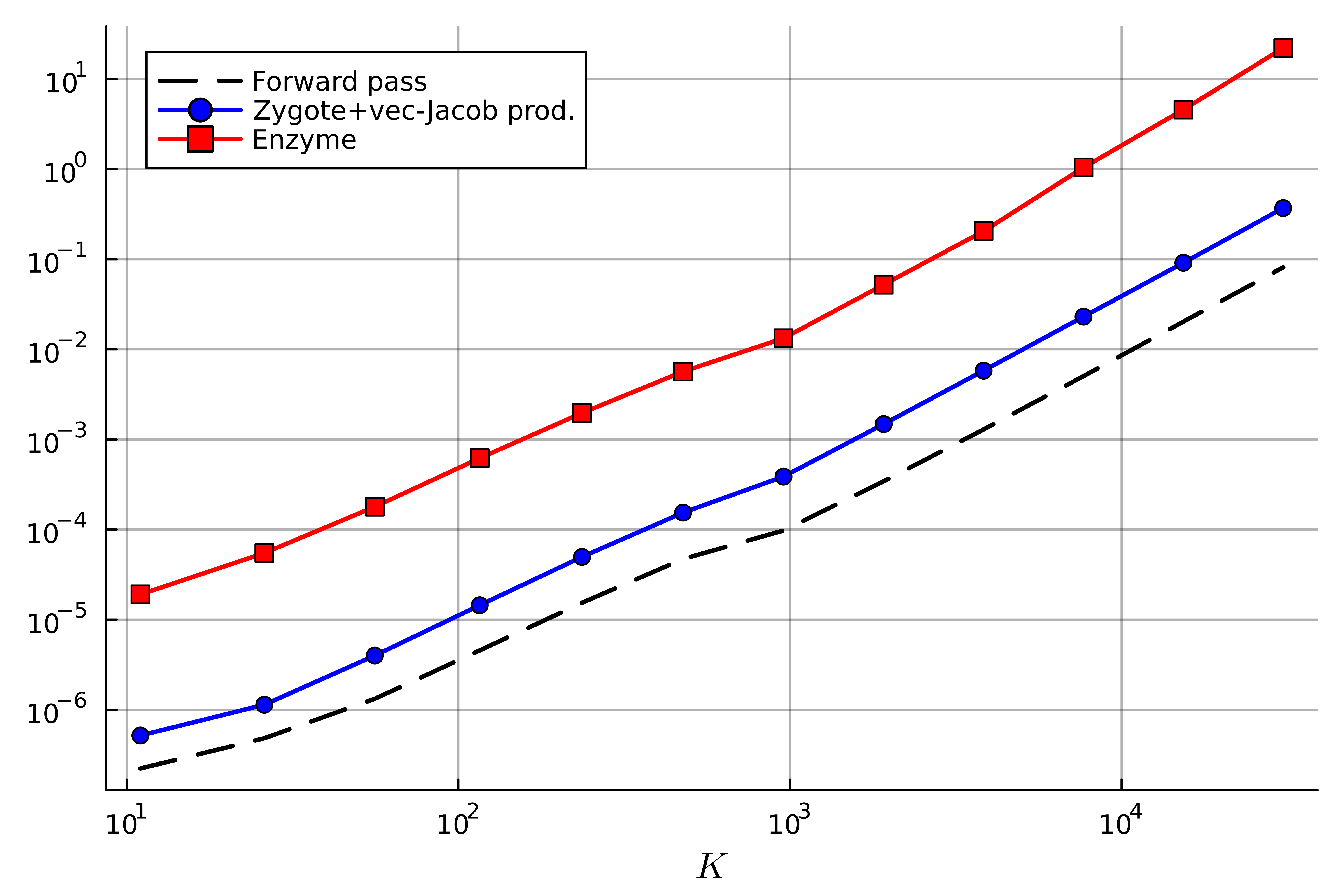}
		\caption{Runtime (s) [acyclic].}
		\label{fig:benchmark_acyclic_time}
	\end{subfigure}%
	\begin{subfigure}{0.45\textwidth}
		\centering
		\includegraphics[scale=0.03]{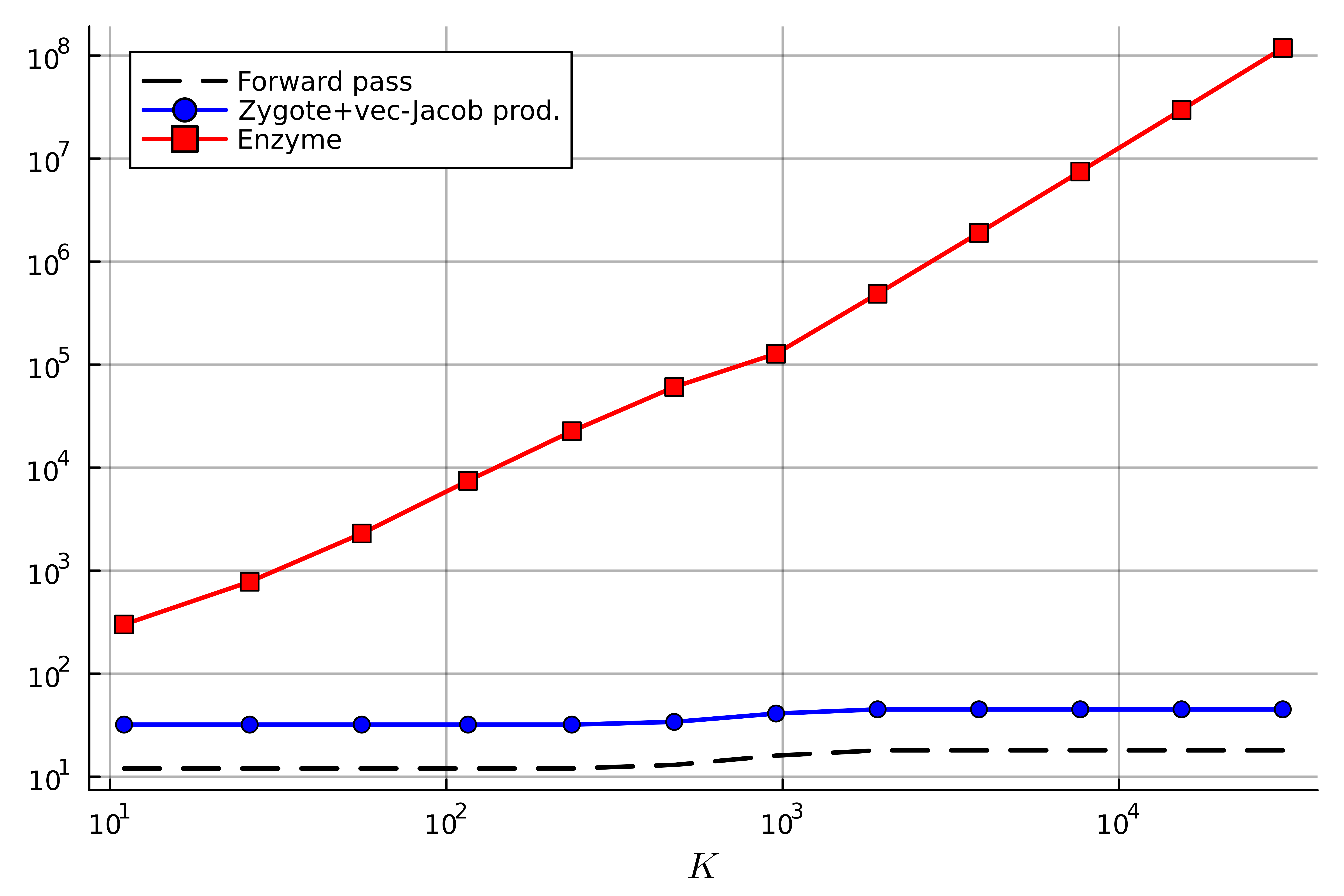}
		\caption{\# Allocations [acyclic].}
		\label{fig:benchmark_acyclic_allocs}
	\end{subfigure}
	\begin{subfigure}{0.45\textwidth}
		\centering
		\includegraphics[scale=0.03]{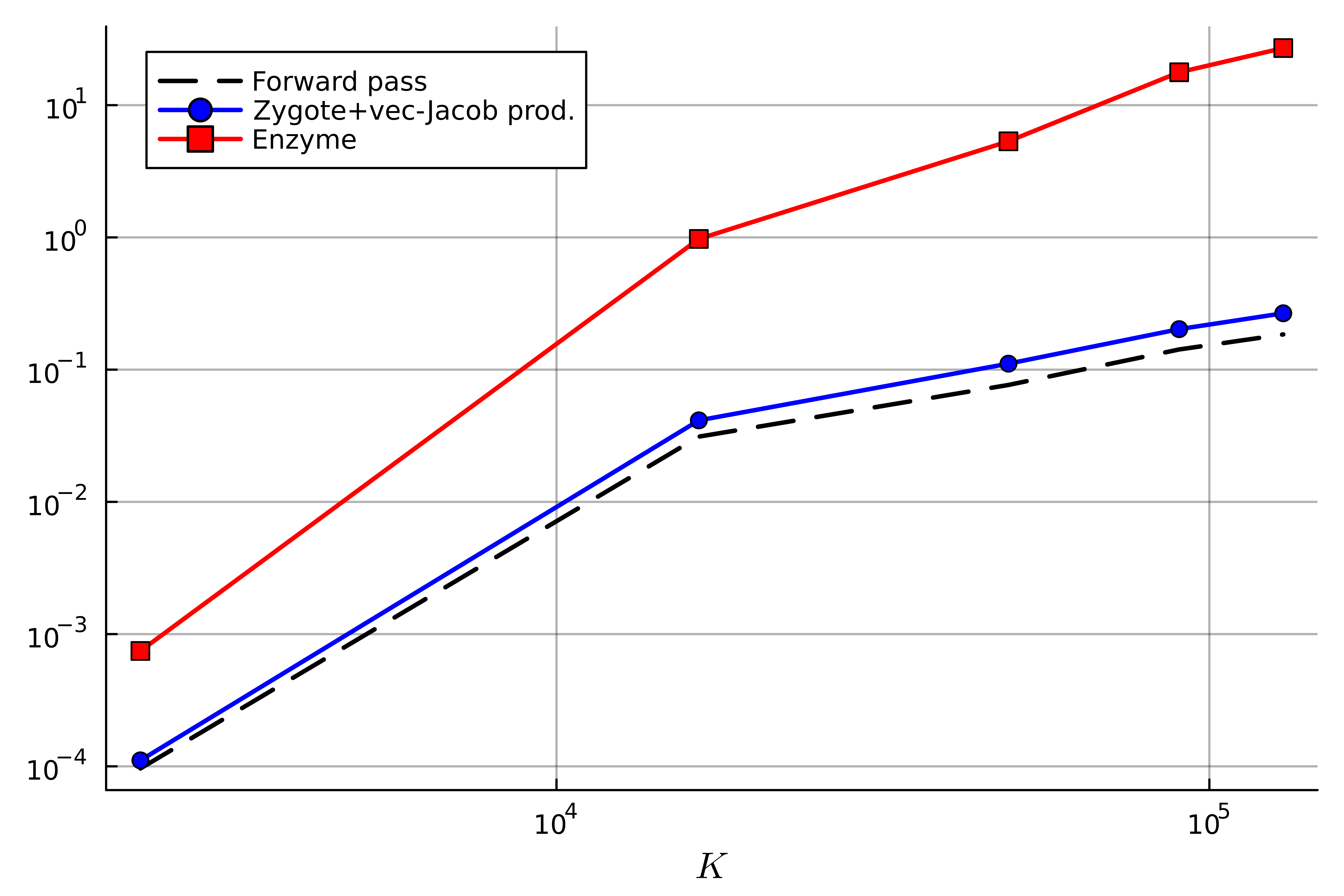}
		\caption{Runtime (s) [cyclic].}
		\label{fig:benchmark_cyclic_time}
	\end{subfigure}%
	\begin{subfigure}{0.45\textwidth}
		\centering
		\includegraphics[scale=0.03]{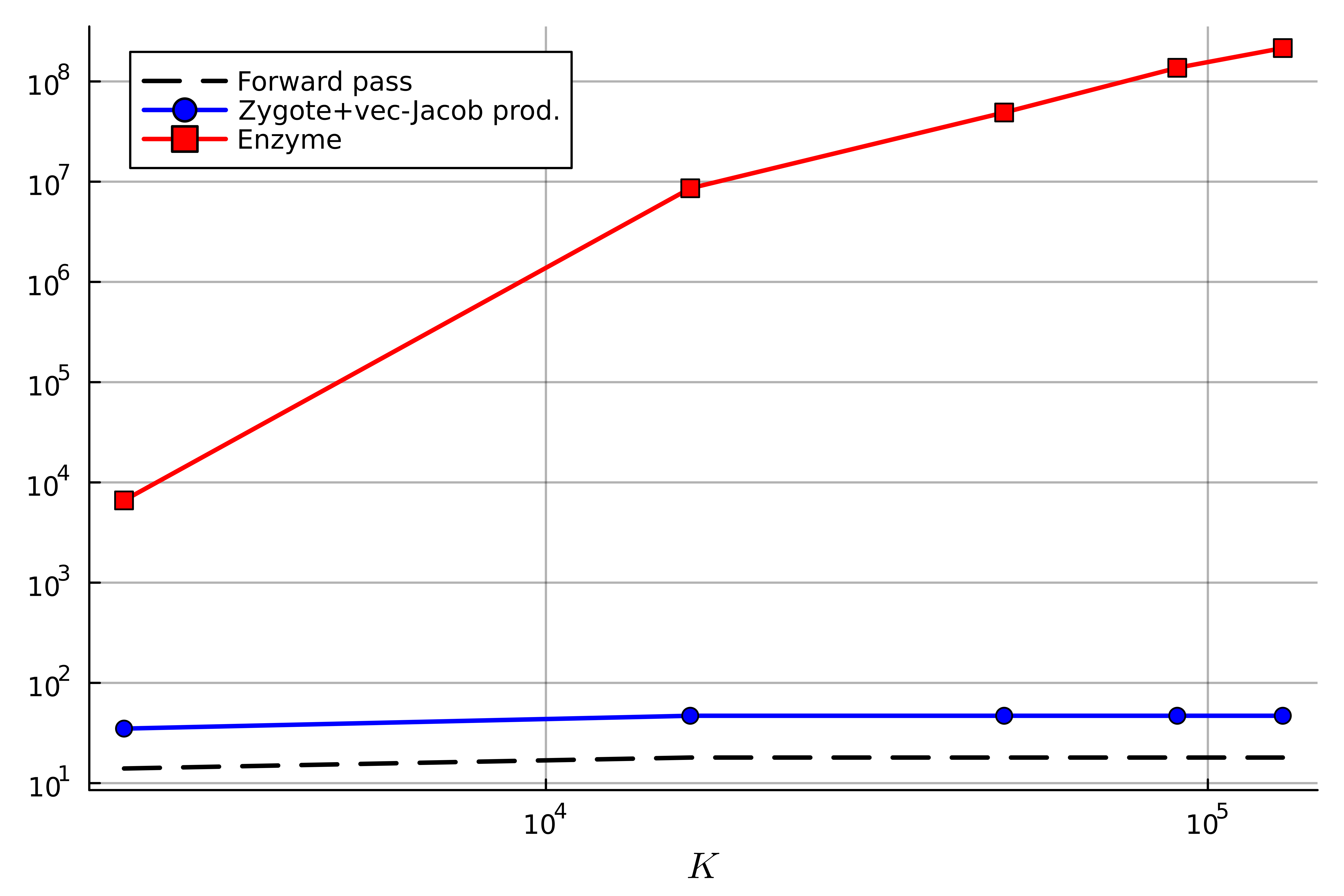}
		\caption{\# Allocations [cyclic].}
		\label{fig:benchmark_cyclic_allocs}
	\end{subfigure}
	\caption{
		Benchmark comparing the AD of $\nu(\mathcal{A})$ in the log-semiring with and without the morphism-trick. 
		The baseline uses Enzyme.jl to perform AD while our proposed method relies on Zygote.jl augmented with a general vector-Jacobian product derived from the morphism-trick.
		The black dashed line correspond to the cost of computing $\nu(\mathcal{A})$, \emph{i.e.} the forward pass, and $K$ is the size of the automaton measured as the number of transitions plus the number of states.
	}
	\label{fig:benchmark}
\end{figure}
To demonstrate the performance of our proposed method, we provide here a benchmark comparing the AD of the function $\nu(\mathcal{A})$, computed in the log-semiring, with and without the morphism-trick. 
We distinguish three cases:
\begin{itemize}
	\item {
		$\mathcal{A}$ is acyclic and topologically sorted:  $\mathcal{A}$ is constructed by iterated concatenations of the automaton in Figure~\ref{fig:fsa_acyclic_topsorted}. 	
		We vary the number of iterated concatenations from $10^0$ to $10^7$.
		An illustration of the concatenation operation is given in Figure~\ref{fig:fsa_acyclic_topsorted_n2}.

	}
	\item {
		$\mathcal{A}$ is acyclic but \emph{not} topologically sorted:  $\mathcal{A}$ is constructed by iterated concatenations of the automaton in Figure~\ref{fig:fsa_acyclic}. 	
		We vary the number of iterated concatenations from $2^0$ to $2^{11}$.
	}
	\item {
		$\mathcal{A}$ is cyclic: $\mathcal{A}$ is an unsmoothed $n$-gram language model estimated on the text of Oscar Wilde \emph{Importance of Being Earnest} using the recipe from OpenGrm NGram Library\footnote{\url{https://www.opengrm.org/twiki/bin/view/GRM/NGramQuickTour}}. 
		We vary the $n$-gram order from $1$ to $5$.
	}
\end{itemize}
For each case, we measure the runtime (in seconds) and the number of memory allocations of the forward and backward pass together. 
We use Enzyme.jl \cite{Moses2020} for the baseline without the morphism-trick, as it presents remarkable performances.
The morphism-trick pipeline uses Zygote.jl augmented with a general vector-Jacobian product implemented in TensorAutomata.jl and ASparseArrays.jl via ChainRules.jl. 
We do not report the results for Zygote.jl alone as they were significantly worse than Enzyme.jl and far less practical due to the difficulty to differentiate in-place memory operations.
The benchmark was evaluated on an Apple M1 chip with 8 Go of memory. 

The results of the benchmark are reported in Figure~\ref{fig:benchmark}.
The case of the acyclic and topologically sorted automaton is the most efficient as the recursion is of finite-length and has a linear complexity.
In this situation, Enzyme.jl peforms nearly an order of magnitude slower than our proposed method.  
The performance gap grows even more when $\mathcal{A}$ is acyclic but not topologically sorted. In this case the recursion is finite but has a quadratic complexity generating even more time consuming memory allocations.  
Finally, the last case ($\mathcal{A}$ is cyclic) is the least favorable for Enzyme.jl, in this situation the recursion is theoretically infinite and approximated numerically by measuring the convergence of the solution. 
The deeper the recursion goes before converging, the deeper the backpropogation tree is and correlates with more memory allocations.
In comparison, in all three cases, our method has very small overhead compared to the cost of the forward pass (dashed line). 
By providing a vector-Jacobian product, we drastically reduce the depth of the backpropagation tree and, consequently, the memory requirements.

\section{Conclusion}


This work proposes a principled method to address the implementation of general vector-Jacobian products for DP algorithms and semiring-based computations in general.
To avoid reimplementing these routines for each possible semiring, it is assumed that the $\oplus$-summations in a chosen semiring can be morphed into a regular summation, allowing simplification in the chain of derivatives. 
We coined our method the ``morphism-trick''.
As a consequence, the backpropagation graph can be ``flattened'' saving many dynamic memory allocations, a major performance bottleneck in AD systems.
To integrate a new semiring into the framework it is sufficient for the end-user to provide implementation of three elementary derivatives listed in~\eqref{eq:generic_drule1}.
Experimental results show that augmenting AD systems with semiring-agnostic implementation of vector-Jacobian products brings several order of magnitude speed up compared to state-of-the-art AD systems.
Our approach is applicable to semirings representing a smooth deformation of the semiring of real/complex such as the log-semiring or the log$_\kappa$ semiring.
It is also applicable to some idempotent semirings (\emph{e.g.} the tropical semiring) but requires in this case to augment the semiring with an extra field to evaluate quantities necessary for the backpropagation.

\begin{acks}
This work was partially funded by the French Government Defense Innovation Agency (convention études et de recherche no. 2022 65 0079), and by the SEMIAMOR (CE23-2023-0005) and InExtenso (ANR-23-IAS1-0004) project grants given by the French National Research Agency (ANR).
We also acknowledge support from the JSALT 2023 workshop, hosted at Le Mans University,
France, and sponsored by Johns Hopkins University with unrestricted
gifts from Amazon, Facebook, Google, and Microsoft.
\end{acks}

\bibliographystyle{ACM-Reference-Format}
\bibliography{refs}

\end{document}